\documentclass[english,12pt]{article}
\usepackage[T1]{fontenc}
\usepackage[latin1]{inputenc}
\usepackage{geometry}
\usepackage{float}
\usepackage{mathrsfs}
\usepackage{amsmath}
\usepackage{amssymb}
\usepackage{graphicx}
\usepackage{tikz}
\usepackage{pgfplots}
\usepackage{bbm}
\usepackage{hyperref}
\hypersetup{
    colorlinks=true,
    linkcolor=blue,
    filecolor=magenta,      
    urlcolor=cyan,
    citecolor = red,
}

\makeatletter

\floatstyle{ruled}
\usepackage{algorithm}
\usepackage{algpseudocode}
\usepackage{subfig}

\usepackage{caption}

\DeclareCaptionFormat{algor}{%
  \hrulefill\par\offinterlineskip\vskip1pt%
    \textbf{#1#2:} #3\offinterlineskip\hrulefill}
\DeclareCaptionStyle{algori}{singlelinecheck=off,format=algor,labelsep=space}
\usepackage[color=orange]{todonotes}

\usepackage{amsthm}

\usepackage{mathrsfs}

\usepackage{amsfonts}

\usepackage{epsfig}

\usepackage{bm}

\usepackage{mathrsfs}

\usepackage{enumerate}

\@ifundefined{definecolor}{\@ifundefined{definecolor}
 {\@ifundefined{definecolor}
 {\usepackage{color}}{}
}{}
}{}

\newtheorem{lem}{Lemma}[section]
\newtheorem{rem}{Remark}[section]
\newtheorem{prop}{Proposition}[section]

\providecommand{\algorithmname}{Algorithm}
\floatname{algorithm}{\protect\algorithmname}

\usepackage[noabbrev, nameinlink]{cleveref}
\crefname{prop}{proposition}{propositions}
\crefname{lem}{lemma}{lemmas}
\crefname{algorithm}{algorithm}{algorithms}
\Crefname{algorithm}{Algorithm}{Algorithms}

\makeatletter
\let\oldtheequation\theequation
\renewcommand\theequation{(\oldtheequation)}
\makeatother

\newcounter{hypA}
\newenvironment{hypA}{\refstepcounter{hypA}\begin{itemize}
  \item[({\bf A\arabic{hypA}})]}{\end{itemize}}
\newcounter{hypB}

\newcounter{hypD}

\usepackage{babel}\date{}

\makeatother

\usepackage{mathtools}

\usepackage{caption}
\numberwithin{equation}{section}

\DeclareCaptionFormat{algor}{%
  \hrulefill\par\offinterlineskip\vskip1pt%
    \textbf{#1#2:} #3\offinterlineskip\hrulefill}
\DeclareCaptionStyle{algori}{singlelinecheck=off,format=algor,labelsep=space}

\usepackage{yhmath}

\begin{document}

%+Title
\begin{center}

{\Large \textbf{Unbiased Estimation using Underdamped \\ Langevin Dynamics}}

\vspace{0.5cm}

 HAMZA RUZAYQAT, NEIL K. CHADA, \& AJAY JASRA

{\footnotesize Applied Mathematics and Computational Science Program, \\ Computer, Electrical and Mathematical Sciences and Engineering Division, \\ King Abdullah University of Science and Technology, Thuwal, 23955-6900, KSA.} \\
{\footnotesize E-Mail:\,} \texttt{\emph{\footnotesize hamza.ruzayqat@kaust.edu.sa, neilchada123@gmail.com, ajay.jasra@kaust.edu.sa}}

\begin{abstract}
In this work we consider the unbiased estimation of expectations w.r.t.~probability measures that have non-negative Lebesgue density, and which are known
point-wise up-to a normalizing constant. We focus upon developing an unbiased method via  the underdamped Langevin dynamics, which has proven to be 
popular of late due to applications in statistics and machine learning. Specifically in continuous-time, the dynamics can be constructed \textcolor{black}{so that as the time goes to infinity they} admit the probability of interest as a stationary measure. \textcolor{black}{In many cases, time-discretized versions of the underdamped Langevin dynamics are used in practice which are run only with a fixed number of iterations.}
We develop a novel scheme based upon doubly randomized estimation as in \cite{ub_grad,disc_model}, which requires access only to time-discretized versions of the dynamics. \textcolor{black}{The proposed scheme aims to remove the dicretization bias and the bias resulting from running the dynamics for a finite number of iterations}. We prove, under standard assumptions, that our estimator is of finite variance and either has finite expected cost, or
has finite cost with a high probability. To illustrate our theoretical findings we provide numerical experiments which verify our theory, which include challenging examples from Bayesian statistics and statistical physics.
\\ \bigskip
\noindent \textbf{Key words}: Underdamped Langevin dynamics, unbiased estimation, maximal couplings, Markov chain simulation \\
\noindent \textbf{AMS subject classifications}: 60J22, 65C05, 65C40, 82C31, 62G08, 35Q56 \\
\noindent\textbf{Code available at:} \url{https://github.com/ruzayqat/unbiased_uld}\\
\noindent\textbf{Corresponding author}: Hamza Ruzayqat. E-mail:
\href{mailto:hamza.ruzayqat@kaust.edu.sa}{hamza.ruzayqat@kaust.edu.sa}

\end{abstract}

\end{center}

\section{Introduction}
\label{sec:intro}

We consider a class of probability measures on the measurable space $(\mathbb{R}^d,\mathcal{B}(\mathbb{R}^d))$ with non-negative Lebesgue densities known point-wise
up-to a normalizing constant. The objective of this article is to consider simulation-based methods, which can return stochastic estimates of finite expectations of functions w.r.t.~the afore-mentioned probability measure, 
that are unbiased, that is, on average are equal to the expectation of interest. This latter task is of interest in a wide variety of areas such as applied mathematics, physics and statistics; see e.g.~\cite{robert} for a book length introduction. 

The primary methodology that is used in the literature to approximate expectations is that of Markov chain Monte Carlo (MCMC) methods. These are schemes which generate ergodic Markov chains whose stationary distribution is exactly the one of interest and there are numerous variants of MCMC, such as random walk Metropolis-Hastings, Hamiltonian Monte Carlo and non-reversible MCMC; see e.g.~\cite{robert}. In addition to this, are methods based upon uncorrected time discretizations of continuous-time processes, which also have the appropriate distribution of interest as a stationary distribution, such as the underdamped Langevin algorithm e.g.~\cite{frenkel,greg}. In particular with the underdamped Langevin algorithm, it has been empirically observed to converge with a better rate, to an invariant distribution, than that of the overdamped Langevin dynamics, which is much simpler in comparison \cite{cheng,dalalyan,durmus}. As a result, these latter methods are of interest often due to their relative ease of simulation relative to the former and have gained significant popularity in the statistics and machine learning literature \cite{gao,fractional}. In both of the classes of methods that we have mentioned, in general, without starting these Markov chains from draws from the target distribution of interest, one seldom returns unbiased estimates which is the main interest in this article. Unbiasedness can be desirable in certain contexts, for instance when computing sensitivities for stochastic gradient algorithms e.g.~\cite{benven}.

We focus upon contributing to the class of unbiased MCMC algorithms and trying to enhance the applicability of such schemes. Unbiased estimation can be achieved in at least two ways, by exact simulation \cite{beskos, beskos2} (e.g.~starting the chain from the correct probability of interest) or unbiased approximation. The former has been investigated many years ago in the guise of coupling from the past MCMC algorithms (e.g.~\cite{propp}), but due to the complications of doing so, the application of such methods is rather rare. The latter is often based upon the pioneering works on unbiased estimation found in \cite{glynn2} (see also \cite{mcl,rhee}). The idea of these latter methods, in the Markov chain context, is to work with a pair of Markov chains on a product space and simulate them until they are equal (the \emph{meeting time}); there is then a novel identity which ensures that the estimate is unbiased. Creating a methodology for allowing the Markov chains to meet was developed in the paper \cite{jacob1} and subsequent to this, several modifications and improvements were given in \cite{midd,boom}. However there has been recent interest of unbiased methodologies, related to that latter way, where we provide some of these works  \cite{chada2,chada,ub_grad,levy,zakai}, \textcolor{black}{which made extensions in the context of particle filtering and particle MCMC methods}.

One of the main issues of the methodology developed in \cite{jacob1} and the sequels, is that one needs to consider sometimes quite complex coupling of pairs of Markov chains.
This can be quite non-trivial to achieve and, in some scenarios such as when the target density is multimodal, rather inefficient, leading to large variances in estimation. This was partially addressed in \cite{sfs} which considered an unbiased version of the Schr\"{o}dinger-F\"{o}llmer sampler (SFS). The latter is a diffusion process on a bounded time domain $[0,1]$, that transports a degenerate distribution at 0, to the target of interest, assuming that the latter is absolutely continuous w.r.t.~a $d-$dimensional standard Gaussian. Even under Euler time-discretization, the process cannot be simulated as the drift term is complicated resulting it in being intractable, but several mechanisms are available; see \cite{sfs}. The authors in that paper show that by using doubly randomized unbiased schemes (e.g.~\cite{ub_bip,ub_pf}) an unbiased version of the SFS can be developed which provides unbiased estimates without having to resort to complex coupling mechanisms. One of the drawbacks of the methodology, however, is the SFS method in the beginning; when approximating the drift, the method can struggle to well-represent complex probability measures. \\
\textcolor{black}{
As a result the focus of this article is on the development of unbiased schemes which alleviate the issues discussed above. Specifically we want to consider unbiased schemes, which can handle two forms
of bias, that from the MCMC and from the discretization bias, arising from models such as stochastic differential equations. This extends the work of  \cite{jacob1} which does not consider the latter as a form
of bias. To do so we exploit ideas randomized multilevel Monte Carlo (MLMC) methods, which have been developed to reduce the cost to attain a particular order of MSE $\mathcal{O}(\epsilon^2), \epsilon>0$.
A recent study showing the connection with the unbiased MC and MLMC is the work of Vihola \cite{vihola}. Also we aim to provide a method which does not require a complex coupling of Markov chains, which is relatively
simple to implement, which has been well understood in theory and practice.
 }
\subsection{Contributions}
The contribution of this article is to develop a new version of the underdamped Langevin algorithm which, even when only working with time-discretizations of the process, can deliver unbiased estimates of expectations w.r.t.~the class of probability measures under considerations. Our motivation is that, like similarly to the work in \cite{sfs}, it does not require very complex coupling techniques and also relies on a double randomization scheme which was developed in \cite{disc_model}. The method also retains the advantages of the unbiased MCMC methods in \cite{jacob1}, that being that the estimates can be computed in an embarrassingly parallel manner, which provides computational speed-ups versus traditional MCMC algorithms.
This is important as our methodology handles two forms of bias, one which arises from the MCMC, and the second arising from the discretization bias associated to a model problem. This is a distinguishment over methods developed by Jacob et al., \cite{jacob1}, but also poses improvements, as we will demonstrate later in the paper, over recently developed methods like the SFS method.
In particular from our work our highlighting contribution is that we establish that our estimate is unbiased and of finite variance and, either has finite expected cost to compute, or has finite cost with high probability. These results rely heavily upon the works in \cite{disc_model} and \cite{disc_lange}, where we assume similar assumptions for the former, and the latter of which provides $\mathscr{V}$-uniform ergodicity of the discretized Markov kernel that we use, where $\mathscr{V}$ is a chosen Lypaunov function. To highlight our theoretical findings we provide numerical experiments on a range of interesting and challenging examples, which arise in statistics and physics. These include a Bayesian logistic regression problem, a double well potential model and finally a Ginzburg-Landau model. We demonstrate both the estimator achieving finite variance and unbiasedness. We compare it to the SFS \cite{sfs} to demonstrate the gains achieved from our proposed estimator.   {To demonstrate the significance of our algorithm further, we also compare it to the unbiased Metropolis adjusted Langevin algorithm (U-MALA) on one of our numerical examples.}

\subsection{Outline}
This article is structured as follows. In \autoref{sec:approach} we detail our proposed methodology based on the underdamped Langevin dynamics. This will lead onto \autoref{sec:theory} where we establish that our estimator is unbiased and of finite variance. Numerical experiments are then conducted in \autoref{sec:numerics}, where we illustrate our methods on several challenging examples. We conclude our remarks and future areas of research in \autoref{sec:conc}. Finally the appendix houses a technical result used in  \autoref{sec:theory}, and the majority of the algorithms introduced.

\section{Approach}
\label{sec:approach}

In this section we first provide a common notation which will be used throughout the manuscript. We then will introduce our approach for unbiased estimation. This will include a review and discussion on the underdamped Langevin dynamics, and in our context. Finally we will discuss how this is related to our new unbiased estimator in the context of maximal couplings and provide our unbiased estimator through various algorithms.

\subsection{Notations}

Let $(\mathsf{X},\mathcal{X})$ be a measurable space.
For $\varphi:\mathsf{X}\rightarrow\mathbb{R}$ we write $\mathcal{B}_b(\mathsf{X})$, 
to denote the collection of bounded measurable functions.
 %and, if $\mathsf{X}\subseteq\mathbb{R}^d$, $\mathbb{L}^2(\mathsf{X})$ as the collection of square Lebesgue-integrable functions. 
 For $\varphi\in\mathcal{B}_b(\mathsf{X})$, we write the supremum norm as $\|\varphi\|=\sup_{x\in\mathsf{X}}|\varphi(x)|$. 
We denote the Borel sets on
$\mathbb{R}^d$ as $B(\mathbb{R}^d)$. The $d-$dimensional Lebesgue measure is written as $dx$.
For a metric $\mathsf{d:\mathsf{X}\times\mathsf{X}}\rightarrow\mathbb{R}^+$ on $\mathsf{X}$ and a function $\varphi:\mathsf{X}\rightarrow\mathbb{R}$,
$\textrm{Lip}_{\mathsf{d}}(\mathsf{X})$ are the Lipschitz functions (with finite Lipschitz constants), that is for every $(x,w)\in\mathsf{X}\times\mathsf{X}$, $|\varphi(x)-\varphi(w)|\leq \|\varphi\|_{\textrm{Lip}}\mathsf{d}(x,w)$.
$\mathscr{P}(\mathsf{X})$  denotes the collection of probability measures on $(\mathsf{X},\mathcal{X})$.
For a finite measure $\mu$ on $(\mathsf{X},\mathcal{X})$
and a $\varphi\in\mathcal{B}_b(\mathsf{X})$, the notation $\mu(\varphi)=\int_{\mathsf{X}}\varphi(x)\mu(dx)$ is used.
For $(\mathsf{X}\times\mathsf{W},\mathcal{X}\vee\mathcal{W})$ a measurable space and $\mu$ a non-negative finite measure on this space,
we use the tensor-product of functions notation for $(\varphi,\psi)\in\mathcal{B}_b(\mathsf{X})\times\mathcal{B}_b(\mathsf{W})$,
$\mu(\varphi\otimes\psi)=\int_{\mathsf{X}\times\mathsf{W}}\varphi(x)\psi(w)\mu(d(x,w))$.
Given a Markov kernel $K:\mathsf{X}\rightarrow\mathscr{P}(\mathsf{X})$ and a finite measure $\mu$, we use the notations
$
\mu K(dx') = \int_{\mathsf{X}}\mu(dx) K(x,dx')
$
and 
$
K(\varphi)(x) = \int_{\mathsf{X}} \varphi(x') K(x,dx'),
$
for $\varphi\in\mathcal{B}_b(\mathsf{X})$. 
The iterated kernel is $K^n(x_0,dx_n) = \int_{\mathsf{X}^{n-1}}\prod_{i=1}^n K(x_{i-1},dx_i)$.
%For a sequence of Markov kernels $K_1,\dots,K_n$ we write 
%$$
%K_{1:n}(x_0,dx_n) = \int_{\mathsf{X}^{n-1}}\prod_{p=1}^n K_p(x_{p-1},dx_p) .
%$$
%For $\mu,\nu\in\mathcal{P}(\mathsf{X})$, the total variation distance 
%is written $\|\mu-\nu\|_{\textrm{tv}}=\sup_{A\in\mathcal{X}}|\mu(A)-\nu(A)|$.
For $A\in\mathcal{X}$, the indicator function is written as $\mathbb{I}_A(x)$. $\mathbb{Z}^+$ is the set of non-negative integers. 
$I_d$ denotes the $d\times d$ identity matrix. The transpose of a vector or matrix $x$ is denoted as $x^T$. We denote min($a,b$) as $a \wedge b$. The operator $\nabla$ denotes the gradient and $\Delta$ the Laplacian, while $\Delta_l$ denotes a time step-size of $2^{-l}$, $l\in \mathbb{Z}^+$.

\subsection{Framework}

We consider the case of underdamped Langevin dynamics:
\begin{eqnarray}
dX_t & = & V_t dt, \label{eq:dyn1}\\
dV_t & = & \left(b(X_t) - \kappa V_t\right)dt + \sigma dB_t, \label{eq:dyn2}
\end{eqnarray}
where $\{B_t\}_{t\geq 0}$ is a standard $d-$dimensional Brownian motion, $(\kappa,\sigma)\in(0,\infty)^2$ are given friction and diffusion coefficients
and $b:\mathbb{R}^d\rightarrow\mathbb{R}^d$ is a typically of gradient form, $b=-\nabla U$. In the latter case there is, under fairly weak assumptions an invariant measure of the process $\{X_t,V_t\}_{t\geq 0}$, with Lebesgue density, $\pi(x,v)$: 
$$
\pi(x,v) \propto \exp\left\{\frac{-\kappa\left(2U(x) + \|v\|^2\right)}{\sigma^2}\right\}.
$$

In practice, one often resorts to time-discretization of the dynamics \eqref{eq:dyn1}-\eqref{eq:dyn2}, for which we focus upon the Euler discretization of step-size $\Delta_l=2^{-l}$, $l\in\mathbb{N}_0$:
\begin{eqnarray}
X_{(k+1)\Delta_l} & = & X_{k\Delta_l} + V_{k\Delta_l}\Delta_l, \label{eq:disc_dyn1}\\
V_{(k+1)\Delta_l} & = & V_{k\Delta_l} + \left(b(X_{k\Delta_l}) - \kappa V_{k\Delta_l}\right)\Delta_l + \sigma\left(B_{(k+1)\Delta_l}-B_{k\Delta_l}\right), \label{eq:disc_dyn2}
\end{eqnarray}
with $k\in\mathbb{N}_0$. Under assumptions (\textcolor{black}{see \cite{disc_lange} Section 2.3}), for $l$ large enough, the discrete-time Markov chain expressed as $\{X_{k\Delta_l},V_{k\Delta_l}\}_{k\in\mathbb{N}_0}$ 
admits a unique invariant measure $\eta_l$ and moreover will converge (in an appropriate sense) to it geometrically quickly. In addition, we should have that $\eta_l$ will converge to $\pi$, by the convergence of Euler approximations. As it will facilitate the approach we are to introduce, we shall modify \eqref{eq:disc_dyn1} to
\begin{equation}
X_{(k+1)\Delta_l} =  X_{k\Delta_l} + V_{k\Delta_l}\Delta_l + \sigma_l\Gamma_{k,l}, \label{eq:disc_dyn3}
\end{equation}
where, $\{\sigma_l\}_{l\in\mathbb{N}_0}$ is any sequence of non-negative and decreasing constants, that converge to zero and for each $l\in\mathbb{N}_0$, $\{\Gamma_{k,l}\}_{k\in\mathbb{N}_0}$ is a sequence of i.i.d.~$d-$dimensional Gaussian random variables of zero mean and covariance matrix the identity multiplied by $\Delta_l$ and that this family of sequences are independent of all other random variables (and of each other). 
As this modification falls under the framework of \cite{disc_lange}, under assumptions, 
 that for $l$ large enough, the discrete-time Markov chain $\{X_{k\Delta_l},V_{k\Delta_l}\}_{k\in\mathbb{N}_0}$, following the dynamics \eqref{eq:disc_dyn2}-\eqref{eq:disc_dyn3},
admits a unique invariant measure $\pi_l$ (which is likely to be different to $\eta_l$) and moreover will converge (in an appropriate sense) to $\pi_l$ geometrically quickly. In addition, we should have that $\pi_l$ will converge to $\pi$, by the convergence of Euler approximations and that the noise in \eqref{eq:disc_dyn3} will disappear.

The approach here is as follows: by using only the dynamics \eqref{eq:disc_dyn2}-\eqref{eq:disc_dyn3}, we will show that one can produce an estimator of $\pi(\varphi\otimes 1)=\int_{\mathbb{R}^{2d}}\varphi(x,v)\pi(x,v)dxdv$ that is unbiased, where $\varphi:\mathbb{R}^{2d}\rightarrow\mathbb{R}$ is $\pi-$integrable. In practice, of course, we will only be interested in the marginal on the $X-$co-ordinate of $\pi$. In other words, we will be interested in unbiasedly estimating the integral $\int_{\mathbb{R}^d} \Psi(x) \tilde{\pi}(x)dx$, with $\Psi(x):\mathbb{R}^d\to \mathbb{R}$ is $\tilde{\pi}-$integrable and $\tilde{\pi}(x)\propto \exp{\{-(2\kappa/\sigma^2) U(x)\}}$ for some differentiable potential function $U$.

\subsection{A Conditional-Type Max-Coupling}\label{sec:cond_max_coup}

In computational statistics a common procedure is to couple various distributions $\mu_1$ and $\mu_2$. Particular examples of this can include independent coupling, or optimal couplings based on the theory of
optimal transport. For this work we make use of maximum couplings \cite{thor}.

To construct our approach, we will need the following simple idea. Suppose that we have two positive Lebesgue densities $\mu_1(x_1,y)$ and $\mu_2(x_2,y)$ on spaces $\mathsf{X}_1\times\mathbb{R}^{2d}$ and $\mathsf{X}_2\times\mathbb{R}^{2d}$, where $\mathsf{X}_1$ and $\mathsf{X}_2$ are two possibly different dimensional sub-spaces of powers of the real line. Suppose also, that we know pointwise the conditional densities, for $j\in\{1,2\}$, $x_j\in\mathsf{X}_j$ fixed:
$$
\mu_j(y|x_j) = \frac{\mu_j(x_j,y)}{\int_{\mathbb{R}^d} \mu_j(x_j,y)dy}.
$$
Our objective is to sample from a coupling of $\mu_1$ and $\mu_2$, so that there is a non-zero probability that the $Y-$co-ordinate can be equal. Let $\check{\mu}$ be any coupling of the marginals $\mu_1(x_1)$ and $\mu_2(x_2)$ and $\overline{\mu}(d(y,y')|x_1,x_2)$ be the maximal coupling of $\mu_1(\cdot|x_1)$ and $\mu_2(\cdot|x_2)$ then one way to achieve our given objective is to sample from the distribution associated to the probability:
$$
\check{\mu}\left(d(x_1,x_2)\right)\overline{\mu}\left(d(y,y')|x_1,x_2\right),
$$
which can be broadly thought of as a posterior. Note that it is straight-forward to show that the marginal of $(y,x_j)$ is $\mu_j$ for $j\in\{1,2\}$. This concept can be straight-forwardly extended to the case of 4 targets, as in \cite[Section 3.2.2.]{disc_model}, using much the same construction.

\subsection{Method}

The method that we pursue is an adaptation of the approach in \cite{disc_model} which was originally developed for unbiased inference for Bayesian inverse problems. The contribution here is that we provide an original coupling construction that one might expect is easier to apply than to conventional Markov kernels such as Metropolis-Hastings. \textcolor{black}{In addition, to verify the unbiasedness, one needs to merge the theory of the underdampled Langevin dynamics \cite{disc_lange}, and unbiased theory developed in \cite{disc_model}}. The former will be discussed later, with various assumptions stated before the main theorem is presented. We will now discuss our unbiased strategy, based on the work
in  \cite{disc_model}.

\subsubsection{Overall Strategy}

Let $l_*\in\mathbb{N}_0$ be given and $\mathbb{P}_L$ be any positive probability mass function on $\mathbb{N}_{l_*}:=\{l_*,l_*+1,\dots\}$. Let $\{\xi_l\}_{l\in \mathbb{N}_{l_*}}$
be any sequence of independent random variables, such that
\begin{eqnarray*}
\mathbb{E}[\xi_{l_*}] & = & \pi_{l_*}(\varphi) \\
\mathbb{E}[\xi_{l}] & = & \pi_{l}(\varphi) - \pi_{l-1}(\varphi) =: [\pi_l-\pi_{l-1}](\varphi) \quad l\in\{l_*+1,l_*+2,\dots\}.
\end{eqnarray*}
Now, let $L$ be a random variable with probability $\mathbb{P}_L$ that is independent of the sequence $\{\xi_l\}_{l\in \mathbb{N}_{l_*}}$ then
\begin{align}
\label{eq:single_term_est}
\widehat{\pi}(\varphi) = \frac{\xi_l}{\mathbb{P}_L(l)},
\end{align}
is an unbiased estimator of $\pi(\varphi)$; see \cite{mcl,rhee} for the initial statement and proof. Moreover, if
\begin{equation}\label{eq:finite_var}
\sum_{l\in\mathbb{N}_{l_*}}\frac{\mathbb{E}[\xi_l^2]}{\mathbb{P}_L(l)} < +\infty,
\end{equation}
the estimator $\widehat{\pi}(\varphi)$ has finite variance. There is also the independent sum-estimator, which can work better than this estimator, but we shall not discuss that for now. The main challenge is then to construct the sequence $\{\xi_l\}_{l\in \mathbb{N}_{l_*}}$. 

Typically, one will run $M \in \mathbb{N}$ independent replicates of \eqref{eq:single_term_est}, where for each replicate $i$, $l_i\sim \mathbb{P}_L$, and then use the average
$$
\left(\widehat{\pi}(\varphi)\right)_{\text{avg}} := \frac{1}{M}\sum_{i=1}^M \left(\widehat{\pi}(\varphi)\right)^{(i)},
$$
where $\left(\widehat{\pi}(\varphi)\right)^{(i)}$ represents the $i$-th independent replicate of the estimate.
\subsubsection{Unbiased Approximation of $\pi_{l_*}(\varphi)$}\label{sec:ub_mcmc}

Throughout the section $l\in\mathbb{N}_{l_*}$ is fixed; although we are interested in the case $l=l_*$ the subsequent exposition holds for any fixed $l$.
We will use a strategy that was developed in \cite{glynn2} (see also \cite{jacob1}) for obtaining the given estimator. The approach is to construct a Markov chain
on the product space $\mathsf{U}^2$, where $\mathsf{U}=\mathbb{R}^{2d}$. From herein, we consider the Markov transition associated to \eqref{eq:disc_dyn2}-\eqref{eq:disc_dyn3}
over a unit time interval. We denote this kernel as $K_l:\mathsf{U}\rightarrow\mathscr{P}(\mathsf{U})$, where $\mathscr{P}(U)$ are the collection of probability measures on measurable space $(\mathsf{U},\mathcal{U})$, $\mathcal{U}$ is the Borel $\sigma-$field on $\mathbb{R}^{2d}$. The reason for this, will become apparent as we continue in our exposition. 

In order to follow the construction in \cite{glynn2} we will need to construct a Markov chain, $\{U_{n,l,}\tilde{U}_{n,l}\}_{n\in\mathbb{N}_0}$, on $(\mathsf{U}^2,\mathcal{U}\vee\mathcal{U})$ so that marginally, at a given discrete time,
new positions of the chain have the kernel $K_l$ as a marginal. More precisely, we need to construct a Markov kernel, $\check{K}_l:\mathsf{U}^2\rightarrow\mathscr{P}(\mathsf{U}^2)$,  so that for any $(u_l,\tilde{u}_l,A)\in\mathsf{U}^2\times\mathcal{U}$:
\textcolor{black}{
\begin{align}
\label{eq:prope}
\int_{A\times\mathsf{U}} \check{K}_l\left((u_l,\tilde{u}_l),d(u_l',\tilde{u}_l')\right) &= \int_A K_l(u_l,du_l'), \\
\nonumber
\int_{\mathsf{U}\times A} \check{K}_l\left((u_l,\tilde{u}_l),d(u_l',\tilde{u}_l')\right) &= \int_A K_l(\tilde{u}_l,d\tilde{u}_l').
\end{align}}

To build the coupling $\check{K}_l$ we shall decompose this into a mixture of two kernels $\check{Q}_l:\mathsf{U}^2\rightarrow\mathscr{P}(\mathsf{U}^2)$ and
$\check{P}_l:\mathsf{U}^2\rightarrow\mathscr{P}(\mathsf{U}^2)$. We will explain why this is required, as the discussion progresses. The simulation of the first kernel $\check{Q}_l$ is
described in \autoref{alg:q_l_sim}. On inspection of 
\autoref{alg:q_l_sim}, it is clear that this samples from a coupling of $K_l$ which is such that if $u=\tilde{u}$ then
$u'=\tilde{u}'$. We note however, that if $u\neq \tilde{u}$, then there is zero probability that $u'=\tilde{u}'$ and being able to achieve the latter is critical in our construction. This motivates our second kernel $\check{P}_l$ whose simulation is given in \autoref{alg:p_l_sim}.  

\begin{rem}
\textcolor{black}{
The kernel $\check{P}_l$ is essentially the same as $\check{Q}_l$ except at the final time step, one has a non-zero probability that $u'=\tilde{u}'$ \emph{irrespective} of $u,\tilde{u}$ and note that if $u=\tilde{u}$ one will always simulate $u'=\tilde{u}'$. Note also, that  it easily follows that $\check{Q}_l$ samples from a coupling of $K_l$; the kernel is a simple application of the method in \autoref{sec:cond_max_coup}.}
\end{rem}
In \autoref{alg:p_l_sim}, step 4.~the maximal coupling can be sampled, for instance, using the maximal coupling algorithm in \cite{thor} or the reflection maximal coupling described in \cite{jacob1}, where the later performs better for high-dimensional models and therefore we adopt it. Then, for some $\alpha\in(0,1)$ fixed we set 
$$
\check{K}_l = \alpha \check{Q}_l + (1-\alpha)\check{P}_l.
$$ 
The reason for including the kernel $\check{Q}_l $ is that it encourages the positions of $u_{n,l},\tilde{u}_{n,l}$ to be close, but can never achieve equality; this latter task can be achieved by $\check{P}_l$ and is why this latter kernel is used.

To initialize the Markov chain, for some initial probability $\nu_l\in\mathscr{P}(\mathsf{U})$ we use the following idea as in \cite{glynn2}. Let $\check{\nu}_l$ be any coupling of $(\nu_l,\nu_l)$ then the initial probability is taken as:
$$
\overline{\nu}_l\left(d(u',\tilde{u}')\right) = \int_{\mathsf{U}^2} \check{\nu}_l\left(d(u,\tilde{u})\right)K_l(u,du')\delta_{\{\tilde{u}\}}(d\tilde{u}'),
$$
which evolves from the Markov kernel $\check{K}_l$ and satisfies the properties discussed through \eqref{eq:prope}.

Recall we are interested in producing an unbiased estimator
of the quantity
\begin{equation}
\pi_l(\varphi)  := \varphi(U_{k,l}) + \sum_{n=k+1}^{\infty}\{\varphi(U_{n,l})-\varphi({U}_{n-1,l})\}.
\end{equation}
However doing so can be challenging, in particular because we have an infinite sum. Therefore what can be done instead, as eluded too from the discussion above,
is to construct another chain $(\tilde{U}_{n,l})$ such that $U_{n,l}=\tilde{U}_ {n,l}$, for $n \geq \tau$, where when the two chains meet, this is known as the
\textit{meeting time} $\tau_l$, which is defined as 
follows
$$
\tau_l := \inf\{n\geq 1: U_{n,l}=\tilde{U}_{n,l}\},
$$
where we require that $\tau_l$ is almost surely finite and that for each $n\geq \tau_l$, $U_{n,l}=\tilde{U}_ {n,l}$ almost surely. Therefore what we can 
do is replace the infinite limit with the meeting time $\tau_l$ with a lag of 1, where now we 
have the following unbiased estimator of $\pi_l(\varphi)$ for any $k\in\mathbb{N}_0$
\begin{equation}\label{eq:stan_est}
\widehat{\pi_l(\varphi)}_k := \varphi(U_{k,l}) + \sum_{n=k+1}^{\tau_{l}-1}\{\varphi(U_{n,l})-\varphi(\tilde{U}_{n,l})\},
\end{equation}
with $\sum_{n=k+1}^{\tau_{l}-1}\{\varphi(U_{n,l})-\varphi(\tilde{U}_{n,l})\} =0$ if $\tau_l-1 < k+1$.

\subsubsection{Assumptions for $\widehat{\pi}_l(\varphi)$}
In order to prove that \eqref{eq:stan_est} is an unbiased estimator, with finite variance and expect cost, 
we summarize these weak assumptions, without going into exact details.
\begin{itemize}
 \item[{1}.] We require \textit{convergence} of the marginal chains, i.e, 
 $$
 \lim_{n \rightarrow \infty} \mathbb{E}[\varphi(U_{n,l})] =\mathbb{E}_{\pi}[\varphi(U)]. 
 $$
 \item[{2}.] The meeting time $\tau_l := \inf\{n\geq 1: U_{n,l}=\tilde{U}_{n,l}\}$ has \textit{geometric tails}
 $$
\mathbb{P}(\tau > n) \leq C\rho^n, \quad \textrm{for} \quad C<\infty, \  \rho \in (0,1).
$$
\item[{3}.] \textit{Faithfulness}: once both chains meet at the meeting time $\tau_l$, then $U_{n,l} = \tilde{U}_{n,l}$ for $n \geq \tau_l$.
 \end{itemize}
We remark that a time-averaged extension is also possible; let 
$(m,k)\in\mathbb{Z}^+\times \mathbb{Z}^+$, with $m\geq k$, then we have
\begin{equation}\label{eq:time_ave}
\widehat{\pi_l(\varphi)}_{T,k,m} := \frac{1}{m-k+1}\sum_{n=k}^m\varphi(U_{n,l}) + \sum_{n=k+1}^{\tau_{l}-1}\Big(1\wedge\frac{n-k}{m-k+1}\Big)\{\varphi(U_{n,l})-\varphi(\tilde{U}_{n,l})\},
\end{equation}
which recovers \eqref{eq:stan_est} in the case $m=k$.

\begin{figure}[h!]
\centering
\includegraphics[width=0.70\textwidth]{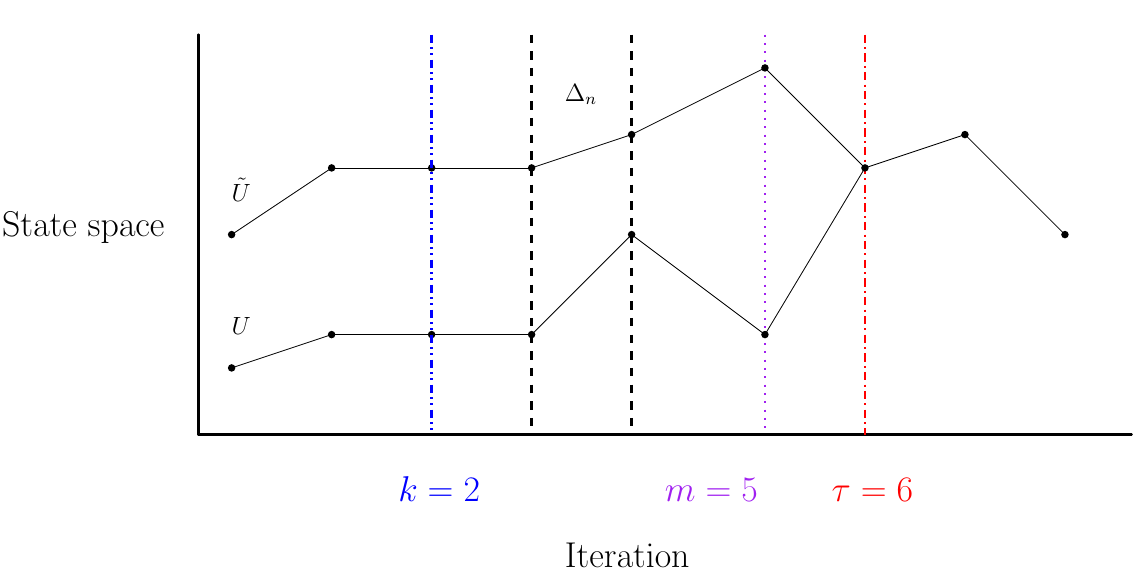}
 \caption{Example illustration of the time-averaged estimator $\widehat{\pi_l(\varphi)}_{T,k,m}$ from \eqref{eq:time_ave}, which includes the meeting time $\tau=7$
with respect to the two chains $U$ and $\tilde{U}$. Each line corresponds to a chain, for which we are aiming to couple.}
 \label{fig:diagram}
\end{figure}

\subsubsection{Unbiased Approximation of $[\pi_l-\pi_{l-1}](\varphi)$}\label{sec:ub_inc_gen}

Our objective is now to provide, for $l\in\{l_*+1,l*+2,\dots\}$ fixed, an unbiased estimator of $[\pi_l-\pi_{l-1}](\varphi)$, which will essentially use the approach detailed in the previous section. Indeed, one could simply use the method outlined above, independently, for $\pi_l$ and $\pi_{l-1}$ and independently for each $l\in\{l_*+1,l*+2,\dots\}$. However, this is unlikely to provide an estimator that can achieve \eqref{eq:finite_var} and hence the variance of such an approach is infinite and not useful in practice. We therefore present an alternative method.
\\\\
The idea we use is to generate a Markov chain on $(\mathsf{Z},\mathcal{Z})$, where $\mathsf{Z}=\mathsf{U}^4$ and $\mathcal{Z}=(\mathcal{U}\vee\mathcal{U})\vee(\mathcal{U}\vee\mathcal{U})$ and we write
$$
Z_{n,l,l-1} = \left((U_{n,l},\tilde{U}_{n,l}),(U_{n,l-1},\tilde{U}_{n,l-1})\right).
$$
The associated Markov kernel $\check{K}_{l,l-1}:\mathsf{Z}\rightarrow\mathscr{P}(\mathsf{Z})$ will be developed below.
The construction of our Markov chain, should be that, marginally $\{U_{n,l}\}_{n\in\mathbb{N}_0}$ (resp.~$\{\tilde{U}_{n,l}\}_{n\in\mathbb{N}_0}$) is a Markov chain with kernel $K_l$
and marginally $\{U_{n,l-1}\}_{n\in\mathbb{N}_0}$ (resp.~$\{\tilde{U}_{n,l-1}\}_{n\in\mathbb{N}_0}$) is a Markov chain with kernel $K_{l-1}$.  
It is explicitly assumed that (at the very least) $(Z_{n,l,l-1})_{n\in\mathbb{N}_0}$ is constructed so that the stopping time $\check{\tau}_{l,l-1}:=\tau_l\vee\tau_{l-1}$ is almost surely finite. In addition, the pair of chains on each level should be faithful, i.e. for $s\in\{l,l-1\}$, we have 
\begin{align}\label{eqn:faithfulness_bothlevels}
U_{n,s}=\tilde{U}_{n,s}, \mbox{ for all } n\geq\tau_{s}.
\end{align}
Hence for time $n\geq\check{\tau}_{l,l-1}$, $Z_{n,l,l-1}$ only has a distinct state on each level. 

Our Markov kernel $\check{K}_{l,l-1}$ will be a type of mixture constituting two distinct Markov kernels $\check{Q}_{l,l-1}:\mathsf{Z}\rightarrow\mathscr{P}(\mathsf{Z})$ and 
$\check{P}_{l,l-1}:\mathsf{Z}\rightarrow\mathscr{P}(\mathsf{Z})$. We describe the simulation of $\check{Q}_{l,l-1}$ in \autoref{alg:q_l_l-1_sim}. This kernel
samples, for $s\in\{l,l-1\}$, from a coupling of $K_s$ which is such that if $u_s=\tilde{u}_s$ then
$u_s'=\tilde{u}_s'$ and if $u_s\neq \tilde{u}_s$, then there is zero probability that $u_s'=\tilde{u}_s'$. The simulation is also such that the outputs across the levels should be quite dependent and, in some sense, close. The additional kernel $\check{P}_{l,l-1}$ that we require is described in \autoref{alg:p_l_l-1_sim}. The kernel $\check{P}_{l,l-1}$ is essentially the same as $\check{Q}_{l,l-1}$ except at the final time step, one has a non-zero probability that, for $s\in\{l,l-1\}$, $u_s'=\tilde{u}_s'$ \emph{irrespective} of $u_s,\tilde{u}_s$ and note that if $u_s=\tilde{u}_s$ one will always simulate $u_s'=\tilde{u}_s'$. Note also, that  it easily follows that $\check{Q}_{l,l-1}$ samples from a couplings of $K_l$ and $K_{l-1}$. In \autoref{alg:p_l_l-1_sim}, step 4.~the synchronous pairwise reflection maximal coupling is described in (\textcolor{black}{\cite{disc_model} Section 3.2.3}) and the simulation thereof is similar to the reflection maximal coupling. Then for $\alpha\in(0,1)$ we set 
$$
\check{K}_{l,l-1}=\mathbb{I}_{D^2}(u_l,\tilde{u}_l,u_{l-1},\tilde{u}_{l-1})\check{Q}_{l,l-1} + \mathbb{I}_{(D^2)^c}(u_l,\tilde{u}_l,u_{l-1},\tilde{u}_{l-1})[\alpha\check{Q}_{l,l-1} + (1-\alpha)\check{P}_{l,l-1}],
$$
where $D=\{(u,\tilde{u})\in\mathsf{U}^2:u=\tilde{u}\}$.

To initialize the Markov chain, for some initial probabilities, $s\in\{l,l-1\}$, $\nu_s\in\mathscr{P}(\mathsf{U})$ we use the following idea. 
We need a kernel $\overline{K}_{l,l-1}:\mathsf{U}^2\rightarrow\mathscr{P}(\mathsf{U}^2)$ described in \autoref{alg:k_l_l-1_sim} for the initialization.
Let $\check{\nu}_s$ be any coupling of $(\nu_s,\nu_s)$ and $\check{\nu}_{l,l-1}$ be a coupling of $(\check{\nu}_l,\check{\nu}_{l-1})$  then the initial probability is taken as:
$$
\overline{\nu}_{l,l-1}\left(dz_{l,l-1}'\right) = \int_{\mathsf{U}^4} \check{\nu}_{l,l-1}\left(dz_{l,l-1}\right)\overline{K}_{l,l-1}\left((u_l,u_{l-1}),d(u_l',u_{l-1}')\right)\delta_{\{\tilde{u}_l,\tilde{u}_{l-1}\}}\left(d(\tilde{u}_l',\tilde{u}_{l-1})\right).
$$
The process is then sampled at subsequent time-points, using the Markov kernel $\check{K}_{l,l-1}$.
\\\\
Finally then one can estimate $[\pi_l-\pi_{l-1}](\varphi)$ as follows, for any $k\in\mathbb{N}_0$:
\begin{equation}\label{eq:stan_est_inc}
\widehat{[\pi_l-\pi_{l-1}](\varphi)}_k := 
\widehat{\pi_l(\varphi)}_k - \widehat{\pi_{l-1}(\varphi)}_k,
%\varphi(X_{k,l}) + \sum_{n=k+1}^{\tau_{l}-1}\{\varphi(X_{n,l})-\varphi(W_{n,l})\} - \Big(
%\varphi(X_{k,l-1}) + \sum_{n=k+1}^{\tau_{l-1}-1}\{\varphi(X_{n,l-1})-\varphi(W_{n,l-1})\}\Big),
\end{equation}
where $\widehat{\pi_s(\varphi)}_k$ is computed using \eqref{eq:stan_est} 
%based on the pair of chains $\{U_{n,s},\tilde{U}_{n,s}\}_{n\in\mathbb{Z}^+}$ on level $s\in\{l,l-1\}$. 
One can also employ time-averaging, for $(m,k)\in\mathbb{N}_0\times \mathbb{N}_0$ satisfying $m\geq k$:
\begin{align}\label{eq:time_ave_inc}
\widehat{[\pi_l-\pi_{l-1}](\varphi)}_{T,k,m} &:= 
\widehat{\pi_l(\varphi)}_{T,k,m} - \widehat{\pi_{l-1}(\varphi)}_{T,k,m},
%\frac{1}{m-k+1}\sum_{n=k}^m\varphi(X_{n,l}) 
%+ \sum_{n=k+1}^{\tau_{l}-1}\Big(1\wedge\frac{n-k}{m-k+1}\Big)\{\varphi(X_{n,l})-\varphi(W_{n,l})\} \notag \\ 
%& -\Big(\frac{1}{m-k+1}\sum_{n=k}^m\varphi(X_{n,l-1}) + \sum_{n=k+1}^{\tau_{l-1}-1}\Big(1\wedge\frac{n-k}{m-k+1}\Big)\{\varphi(X_{n,l-1})-\varphi(W_{n,l-1})\}\Big).
\end{align}
where $\widehat{\pi_s(\varphi)}_{T,k,m}$ is computed using \eqref{eq:time_ave} 
%based on the pair of chains $\{U_{n,s},\tilde{U}_{n,s}\}_{n\in\mathbb{N}_0}$ on level $s\in\{l,l-1\}$. 

\subsubsection{Final Methodology and Estimator}

\textcolor{black}{We now consolidate the above discussion by summarizing our proposed methodology to unbiasedly estimate $\pi(\varphi)$, which is presented in Algorithm
\ref{alg:final}. We also provide a summary of our proposed methodology, through a simple diagram given in Figure \ref{fig:kernel}.}
\begin{figure}[h!]
\centering
\includegraphics[width=0.55\textwidth]{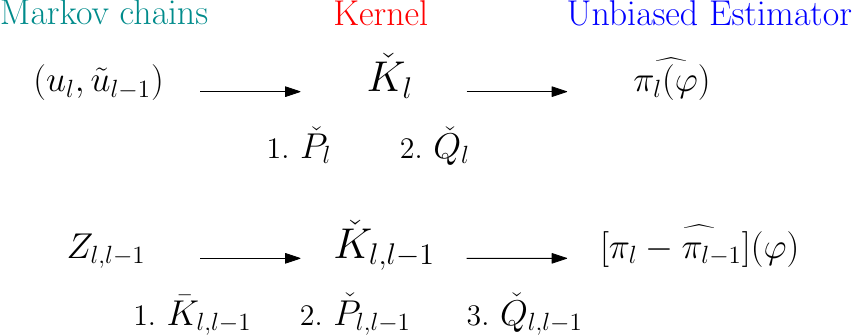}
 \caption{\textcolor{black}{Cartoon description of the required procedure to obtain unbiased estimators, related to the various kernels required.}}
 \label{fig:kernel}
\end{figure} 
%We begin with the single term estimator $\widehat{\pi(\varphi)}_S$ in \eqref{eq:single_term}.
%The procedure that we consider is then as follows, for the case of a single-term estimator.

\begin{algorithm}[!h]
\textcolor{black}{
\caption{Unbiased estimator $\widehat{\pi(\varphi)}$.}
\label{alg:final}
{\bf Input}: Probability mass function $\mathbb{P}_L$, \\ Initialized Markov chains $(u_0,\tilde{u}_0)=\left((x_{0},v_{0}),(\tilde{x}_{0},\tilde{v}_{0})\right), \quad (u_{0,s},\tilde{u}_{0,s})=\left((x_{0,s},v_{0,s}),(\tilde{x}_{0,s},\tilde{v}_{0,s})\right)$.
\begin{enumerate}
\item{Sample $L\sim\mathbb{P}_L$.}
\item{If $L=l_*$, set $\tilde{U}_{0,l_*} = \tilde{u}_0$ sampling $U_{0,l_*}$ through \eqref{eq:disc_dyn2} - \eqref{eq:disc_dyn3} at level $l_*$ up to time 1. Generate 
$\{U_{n,l},\tilde{U}_{n,l}\}_{n\in\mathbb{N}}$ according to $\check{K}_{l_*}$,  
and compute
$\widehat{\pi_0(\varphi)}_k$ in \eqref{eq:stan_est} or 
$\widehat{\pi_0(\varphi)}_{T,k,m}$ in \eqref{eq:time_ave}. }
\item{If $L>l_*$, set $(\tilde{U}_{0,L}, \tilde{U}_{0,L-1}) = (\tilde{u}_{0,F},\tilde{u}_{0,C})$ sampling $(U_{0,L}, U_{0,L-1})$ from $\overline{K}_{L,L-1}$ by running \autoref{alg:k_l_l-1_sim}. Generate $\{Z_{n,L,L-1}\}_{n\in\mathbb{N}}$ according to 
$\check{K}_{L,L-1}$, and compute 
$\widehat{[\pi_L-\pi_{L-1}](\varphi)}_k $ in \eqref{eq:stan_est_inc} or 
$\widehat{[\pi_L-\pi_{L-1}](\varphi)}_{T,k,m}$ in \eqref{eq:time_ave_inc}.
}
\end{enumerate}
{\bf Output}:  Single-term estimator
\begin{equation}\label{eq:basic_ub_est}
\widehat{\pi(\varphi)}_{S,k} := \frac{1}{\mathbb{P}_L(L)}\Big(\mathbb{I}_{\{l_*\}}(L)~\widehat{\pi_{l_*}(\varphi)}_k + \mathbb{I}_{\{l_*+1,l_*+2,\dots\}}(L)~\widehat{[\pi_L-\pi_{L-1}](\varphi)}_k \Big),
\end{equation}
or time-averaged estimator
\begin{equation}\label{eq:basic_ub_est_time}
\widehat{\pi(\varphi)}_{S,T,k,m} := \frac{1}{\mathbb{P}_L(L)}\Big(\mathbb{I}_{\{l_*\}}(L)
~\widehat{\pi_{l_*}(\varphi)}_{T,k,m} + \mathbb{I}_{\{l_*+1,l_*+2,\dots\}}(L)~
\widehat{[\pi_L-\pi_{L-1}](\varphi)}_{T,k,m} \Big).
\end{equation}}
%depending on which estimator we choose.
\end{algorithm}

\section{Theoretical Properties}\label{sec:theory}

In this section we provide theory related to the discussed methodology in \autoref{sec:approach}. Specifically we provide theoretical justification,
where we provide our main result which states that our estimator has finite variance, and either has finite expected cost, or has finite cost with a high probability. 
In order to prove this we require a number of standard assumptions, which have been used in previous works. We omit the proof of theorem in this section, where we defer it to the Appendix.
\textcolor{black}{For the below algorithm we have $s\in \{F,C\}$, where $F$ denotes fine, and $C$ denotes coarse, in terms of the level of discretization.}

\subsection{Assumptions}
Before discussing our required assumptions, we provide some new common notation in the setting of Lyapunov functions.
Let $0<C<+\infty$ be a constant and $\mathsf{d}$ a metric on $\mathsf{U}$; we define the set
\begin{eqnarray*}
\mathsf{B}(C,\Delta_l,\mathsf{d})  & := &  \{u_{1:4}\in\mathsf{U}^4:\forall(i,j)\in\{1,\dots,4\}, \mathsf{d}(u_i,u_j)\leq C\Delta_l\}. 
%B(C,\Delta_l,\mathsf{d})  & := &  \{x_{1:2}\in\mathsf{X}^2:\mathsf{d}(x_1,x_2)> C\Delta_l\}.
\end{eqnarray*}
In what follows, for any $l\in\mathbb{N}_{l_*}$, under assumptions (\textcolor{black}{see \cite{disc_lange} Section 2.3}) the Markov kernel $K_l$ admits a unique invariant measure $\pi_l$.
We introduce a Lyapunov function $\mathscr{V}:\mathsf{U}\rightarrow[1,\infty)$ which will be used in our assumptions below. For $f:\mathsf{U}\rightarrow\mathbb{R}$
and a given Lyapunov function $\mathscr{V}$, we define the collection of functions $\mathcal{L}_{\mathscr{V}}:=\{f:\sup_{u\in\mathsf{U}}|f(u)|/\mathscr{V}(u)<+\infty\}$
and we write $\|f\|_{\mathscr{V}}=|f(u)|/\mathscr{V}(u)$.

\begin{hypA}\label{ass:1}
\begin{enumerate}
\item{There exist a $\mathscr{V}:\mathsf{U}\rightarrow[1,\infty)$, $(\lambda,b,\mathsf{C})\in(0,1)\times(0,\infty)\times\mathcal{U}$ such that for any $(l,u)\in\mathbb{N}_{l_*}\times\mathsf{U}$
$$
K_l(\mathscr{V})(u) \leq \lambda \mathscr{V}(u) + b\mathbb{I}_{\mathsf{C}}(u).
$$
}
\item{There exists a $(\rho,C)\in(0,1)\times(0,\infty)$ such that for any $n\in\mathbb{N}$
$$
\sup_{l\in \mathbb{N}_{l_*}}\sup_{u\in\mathsf{U}}\sup_{\varphi\in\mathcal{L}_{\mathscr{V}}}\frac{|K_l^n(\varphi)(u)-\pi_l(\varphi)|}{\mathscr{V}(u)} \leq C\|\varphi\|_{\mathscr{V}}\rho^n,
$$
with $\mathscr{V}$ as in 1..
}
\end{enumerate}
\end{hypA}
\begin{hypA}\label{ass:2}
There exist $(C,\rho)\in(0,\infty)\times(0,1)$ such that: 
\begin{enumerate}
\item{For any $n\in\mathbb{N}$
$$
\mathbb{E}[\mathbb{I}_{\{\tau_0> n\}}]\leq C\rho^n.
$$}
\item{For any $(l,n,z)\in\mathbb{N}\times\mathbb{N}\times\mathsf{X}^4$
$$
\mathbb{E}[\mathbb{I}_{\{\tau_l> n\}}|Z_{0,l,l-1}=z]\leq C\rho^n.
$$
}
\end{enumerate}
\end{hypA}
\begin{hypA}\label{ass:3}
There exist a $C<\infty$  and metric on $\mathsf{X}$, $\tilde{\mathsf{d}}:\mathsf{X}\times\mathsf{X}\rightarrow\mathbb{R}^+$ such that for any \\ $(l,\varphi,(x,y))\in\mathbb{Z}^+\times\mathcal{B}_b(\mathsf{X})\cap\textrm{Lip}_{\tilde{\mathsf{d}}}(\mathsf{X})\times\mathsf{X}\times\mathsf{X}$
$$
|K_l(\varphi)(x)-K_l(\varphi)(y)| \leq C (\|\varphi\|\vee \|\varphi\|_{\textrm{Lip}})\tilde{\mathsf{d}}(x,y).
$$
\end{hypA}
\begin{hypA}\label{ass:4}
There exist $(C,\beta_1)\in(0,\infty)\times(0,\infty)$ such that for any $(l,\varphi)\in\mathbb{N}\times\mathcal{B}_b(\mathsf{X})$
\begin{enumerate}
\item{
$
|[\pi_l-\pi](\varphi)| \leq C \|\varphi\| \Delta_l^{\beta_1}.
$
}
\item{
$
\sup_{x\in\mathsf{X}}| K_l(\varphi)(x)-K_{l-1}(\varphi)(x) | \leq C  \|\varphi\| \Delta_l^{\beta_1}.
$
}
\end{enumerate}
\end{hypA}
\begin{hypA}\label{ass:5}
\begin{enumerate}
\item{There exist $(C,\epsilon,\beta_2)\in(0,\infty)^3$, such that for the metric  $\tilde{\mathsf{d}}$ as in (A\ref{ass:3}) and any $(l,n)\in\mathbb{N}\times\mathbb{N}$
\begin{eqnarray*}
\mathbb{E}[\mathbb{I}_{\mathsf{B}(C,\Delta_l^{\beta_2},\tilde{\mathsf{d}})^c}(Z_{0,l,l-1})] & \leq & C\Delta_l^{\beta_2(2+\epsilon)}, \\
%\mathbb{E}[\mathbb{I}_{B(C,\Delta_l^{\beta_2},\tilde{\mathsf{d}})^c}(\widetilde{X}_{0,l},\widetilde{X}_{0,l-1})] & \leq & C\Delta_l^{\beta_2}\\
\mathbb{E}[\mathbb{I}_{\mathsf{B}(C,\Delta_l^{\beta_2},\tilde{\mathsf{d}})^c\times B(C,\Delta_l^{\beta_2},\tilde{\mathsf{d}})}(Z_{n,l,l-1},Z_{n-1,l,l-1})] & \leq & C\Delta_l^{\beta_2(2+\epsilon)}.
%\mathbb{E}[\mathbb{I}_{B(C,\Delta_l^{\beta_2},\tilde{\mathsf{d}})\times B(C,\Delta_l^{\beta_2},\tilde{\mathsf{d}})^c}(\widetilde{X}_{n,l:l-1},\widetilde{X}_{n-1,l:l-1})] 
%& \leq & C\Delta_l^{\beta_2}.
\end{eqnarray*}}
\item{$\tilde{d}^{4(2+\epsilon)}\in\mathcal{L}_{\mathscr{V}\otimes\mathscr{V}}$, with $\tilde{\mathsf{d}},\epsilon$ as in 1.~and $\mathscr{V}$ as in (A\ref{ass:1}).}
\end{enumerate}
\end{hypA}

\subsection{Discussion of Assumptions}
\textcolor{black}{
Let us briefly discuss each assumption in-turn. Assumption (A\ref{ass:1}) is related to the ergodicity of the underdamped Langevin sampler.
The specific result we have used, arises in the work \cite{disc_lange}, where the authors show that one can attain $\mathcal{V}$-uniform ergodicity, for
the discretized set of equations. These conditions hold so long as $l_*$ is large enough, which we shall also assume. 
Assumption (A\ref{ass:2}) has been considered in \cite{ub_grad}, and is shown to hold in a related context [\cite{ub_grad}, Lemmata 14 and 20].
Assumption (A\ref{ass:3})-(A\ref{ass:4}) relates to the continuity properties of the kernel and the discretization bias of the problem. In particular (A\ref{ass:4})
2. states that moves under pairs of kernels at consecutive levels, stay close on average,
given that they are initialized at the same point.
Assumption (A\ref{ass:5}) 1. is a non-standard assumption which can be verified in complex settings, i.e. [\cite{ub_grad}, Lemma 16].
Also if one can establish that the pairs of chains are uniformly ergodic, with an invariant measure $\check{\pi}_{l,l-1}$, then it is
reasonable to assume that
$$
\check{\pi}_{l,l-1}\Big((\varphi \otimes1-1\otimes \varphi)^2\Big) \leq C\Delta_l^{2\beta_2}.
$$
The other key assumption which is different in our work is (A\ref{ass:5})
2.~which is a simple integrability assumption on the metric $\tilde{\mathsf{d}}$ which allows one to remove the need for $\mathsf{U}$ to be compact in 
\cite{disc_model} and is not the case in our context.}

\subsection{Main Result and Implications}

We are now in a position to state our main result, which is that our estimator is unbiased with finite variance.  

\begin{prop}\label{prop:main_prop}
Assume (A\ref{ass:1}-\ref{ass:4}). Then there exists a choice of positive probability mass function $\mathbb{P}_L$, such that for 
the metric $\tilde{\mathsf{d}}$ in (A\ref{ass:3}) and any $\varphi\in \mathcal{B}_b(\mathsf{X})\cap\textrm{\emph{Lip}}_{\tilde{\mathsf{d}}}(\mathsf{X})$,  
\eqref{eq:basic_ub_est} is an unbiased and finite variance estimator of $\pi(\varphi)$. 
\end{prop}

\begin{rem}
\label{rem:1}
As in \cite[Theorem 2.1.]{disc_model}, one can easily extend this result to the case of \eqref{eq:basic_ub_est_time}.
\end{rem}

\autoref{prop:main_prop} establishes that there exists a choice of positive probability mass function $\mathbb{P}_L$ which ensures that 
\eqref{eq:basic_ub_est} is unbiased and finite variance, but not how to find one. That is the topic of the current discussion and follows very much that given in \cite[Section 2.4.3.]{disc_model}. Under our assumptions, noting \eqref{eq:finite_var}, using the argument contained in the Appendix and in \cite{disc_model}, the variance of \eqref{eq:basic_ub_est} is upper-bounded by (with $C<\infty$ a generic constant that does not depend upon $l$, but whose value may change on each appearance) 
\begin{align}
\label{eq:bound_cost}
C\sum_{l\in\mathbb{N}_{l_*}}\frac{\mathbb{E}\left[\xi_l^2\right]}{\mathbb{P}_L(l)} ~\leq ~ C\sum_{l\in\mathbb{N}_{l_*}}\frac{\Delta_l^{2\beta}}{\mathbb{P}_L(l)},
\end{align}
with $\xi_l := \widehat{[\pi_l-\pi_{l-1}](\varphi)}_k$ defined in \eqref{eq:stan_est_inc} and $\beta=\min\left\{\beta_1,\tfrac{\beta_2}{2}\right\}$. The expected cost is 
upper-bounded by 
$$
C\sum_{l\in\mathbb{N}_{l_*}}\mathbb{E}\left[\check{\tau}_{l,l-1}\right]\Delta_l^{-1}\mathbb{P}_L(l)~ \leq ~C\sum_{l\in\mathbb{N}_{l_*}}\Delta_l^{-1}\mathbb{P}_L(l),
$$
where it is assumed that the cost to sample from the kernel $K_l$ is bounded by $C\Delta_l^{-1}$ and that, by A\ref{ass:2}, one can upper-bound the expected stopping time $\mathbb{E}\left[\check{\tau}_{l,l-1}\right]$. If one took $\mathbb{P}_L(l)\propto\Delta_l^{\omega}$, then if $\omega\in (1, 2\beta)$ both the variance and expected cost are finite. From standard results on Euler discretization of diffusions one might expect that $\beta_1 = 1$ as confirmed in \autoref{fig:EM_weak_err_DW}, where we estimate the  {weak} error rate for Euler--Maruyama method considering the models in \autoref{sec:numerics}. We also plot the second moment of the increments $\xi_l$ in \autoref{fig:inc_var}. If we assume that $\mathbb{E}[\xi_l^2]=\mathcal{O}(\Delta_l^{2\beta})$, then \autoref{fig:inc_var} gives an estimate of $\beta$ for each model, which is approximately $1$, and as a result, $\beta_2 > 2$.

\begin{figure}[h!]
\centering
\includegraphics[scale=0.16]{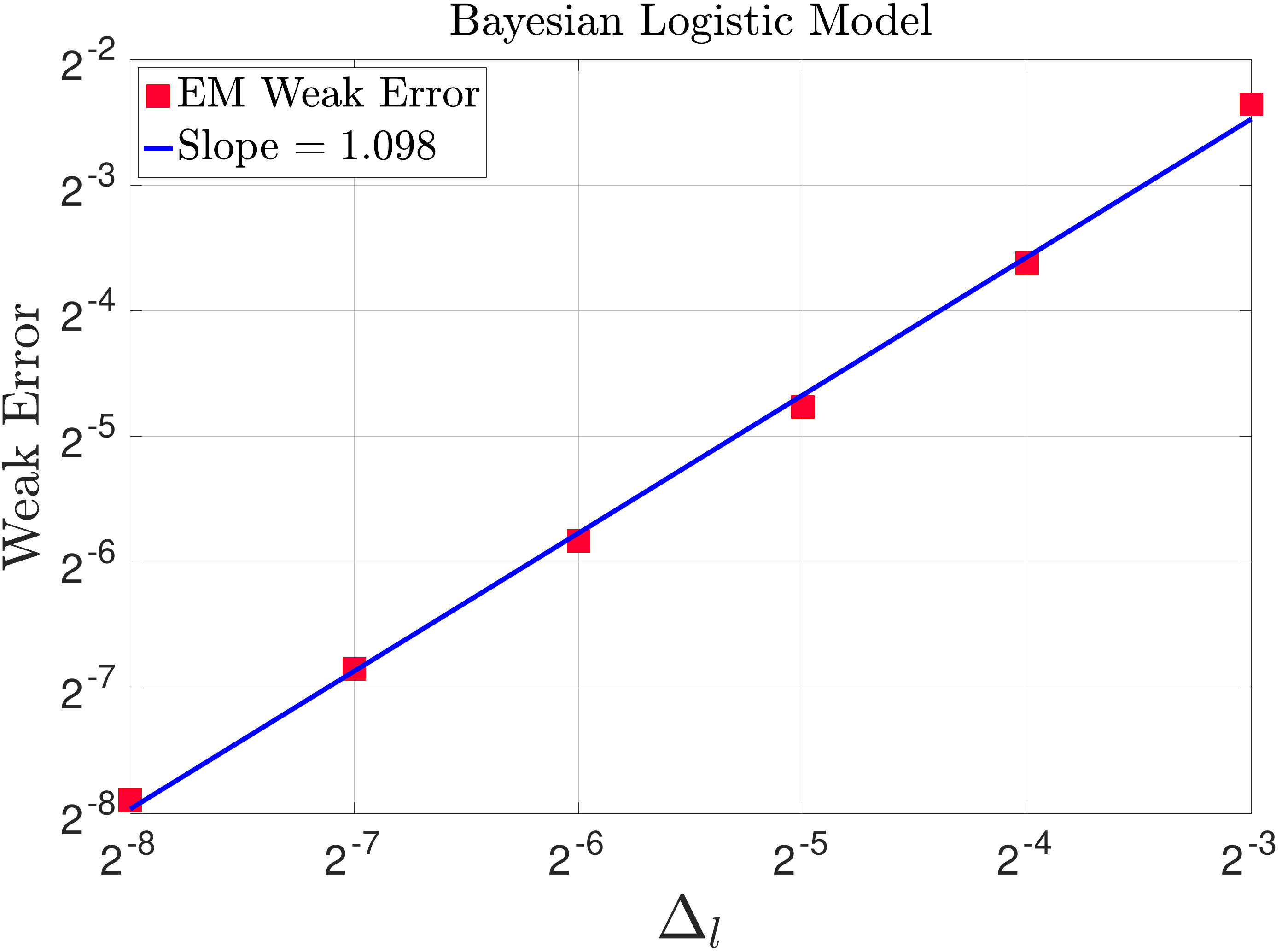} \vspace{10pt}
\includegraphics[scale=0.16]{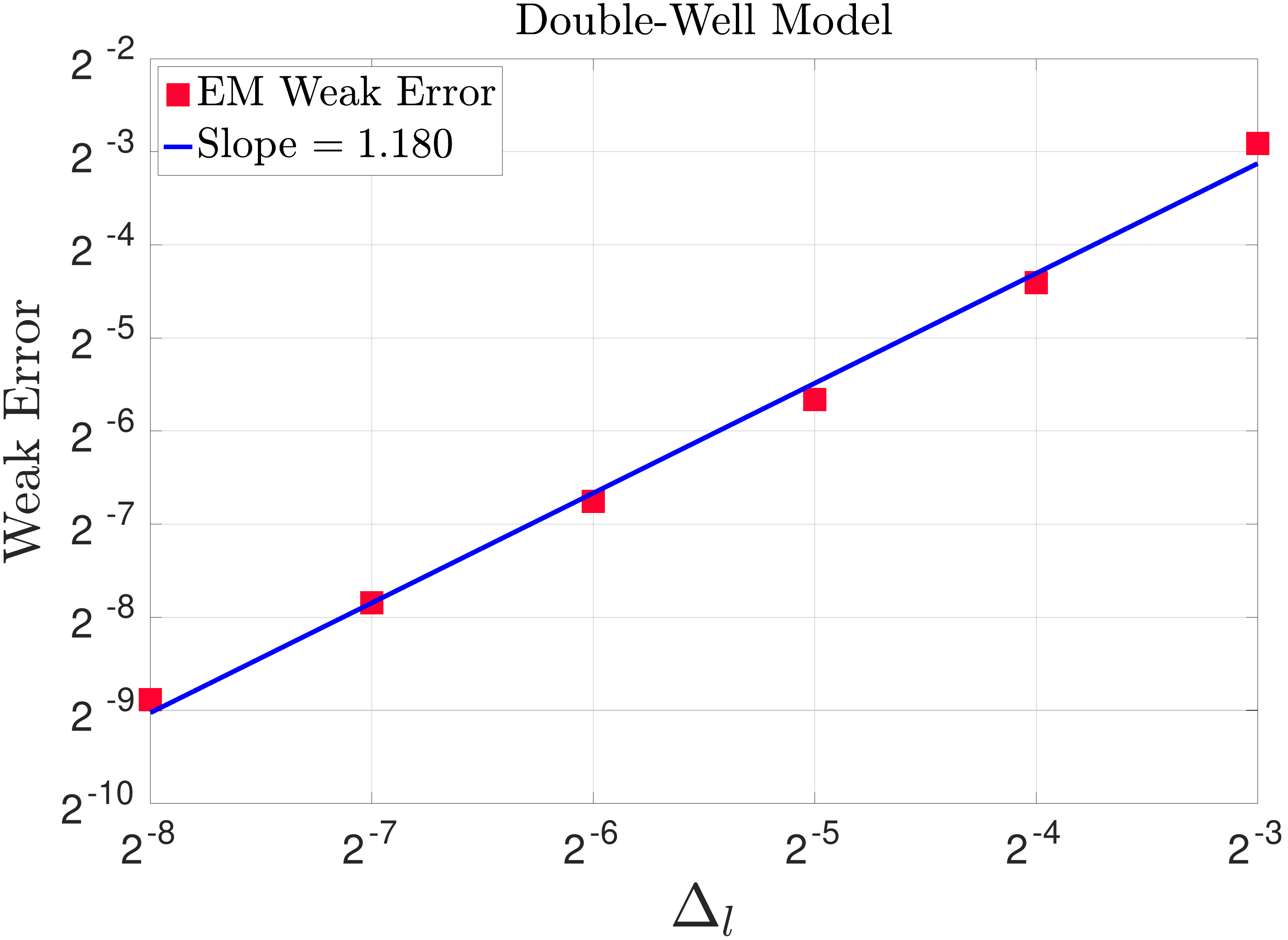}
\includegraphics[scale=0.16]{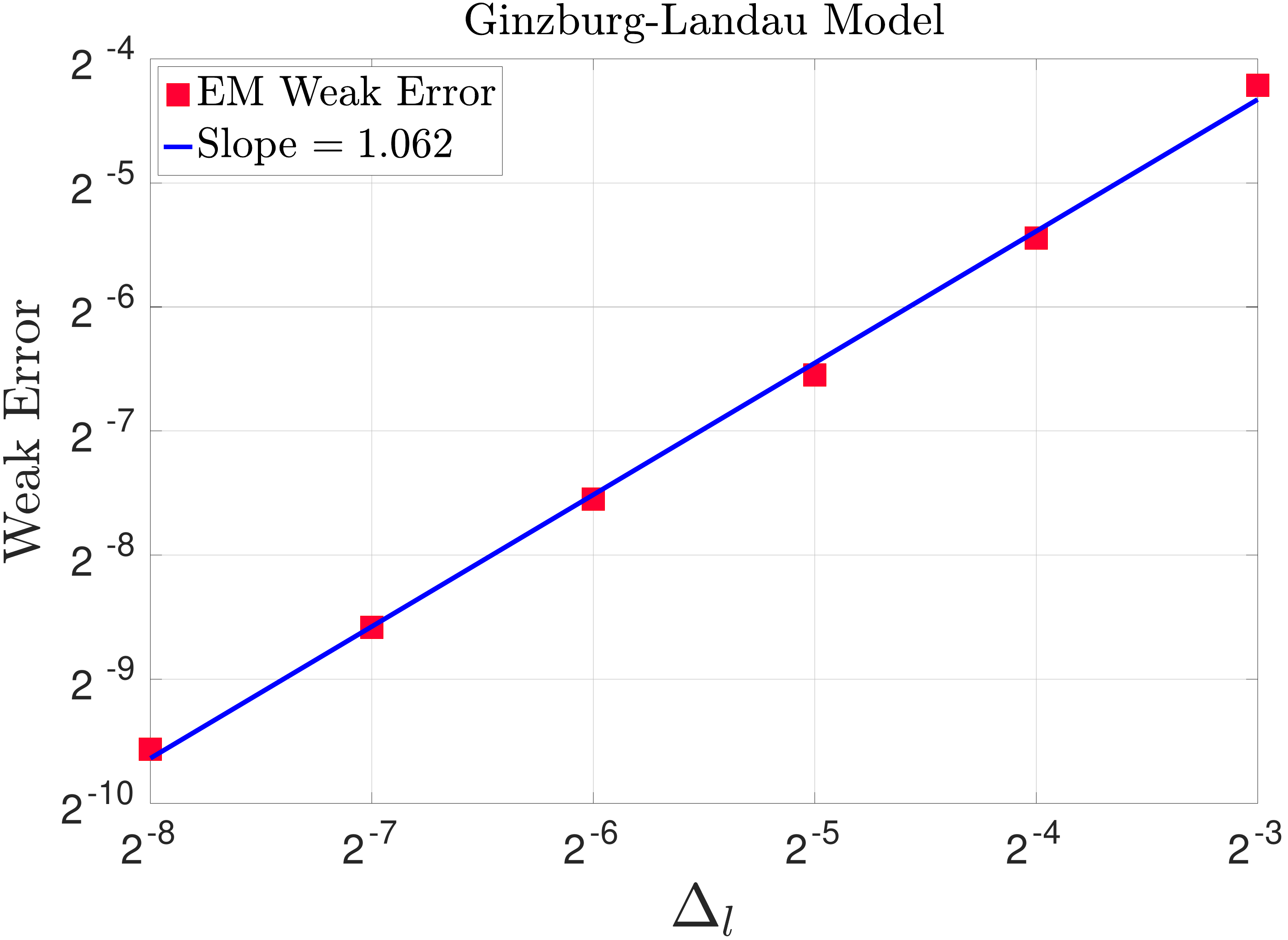}
\captionof{figure}{Euler--Maruyama (E--M) weak error $ |[\pi_l - \pi](\varphi)|$ against the time step-size $\Delta_l$ for the models in \autoref{sec:numerics}, with $l$ the discretization level, $\varphi(u) = u_1$ (the first co-ordinate), $\pi(\varphi)$ the reference expectation which we take to be the average of $5200$ samples of Euler-discretized $\pi_L(\varphi)$ at level $L=18$, a very high level of discretization. These plots verify that the weak error rate $\beta_1$ of E--M is approximately 1 as expected.}
\label{fig:EM_weak_err_DW}
\end{figure}

\begin{figure}[!h]
\centering
\includegraphics[scale=0.16]{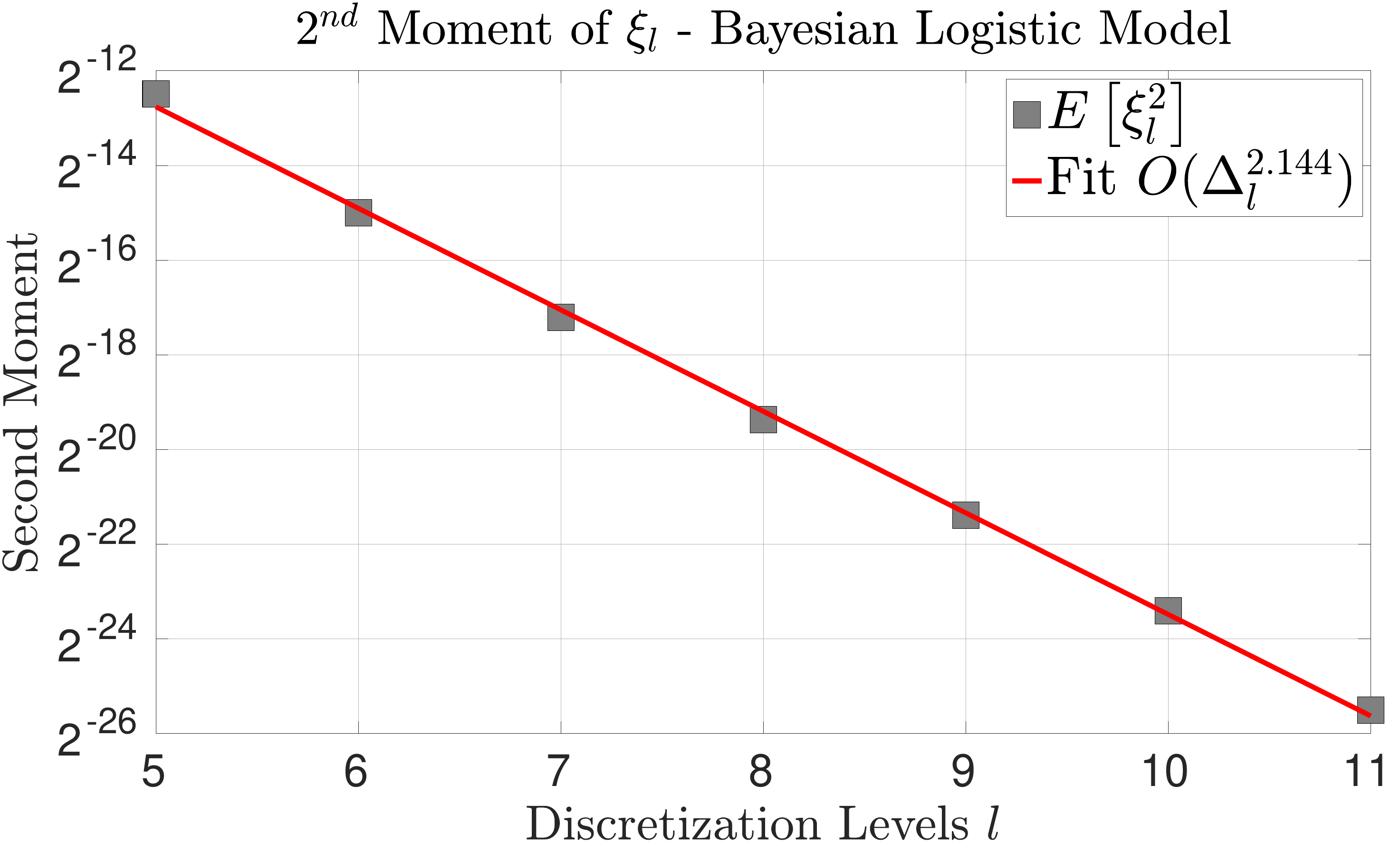}\vspace{10pt}
\includegraphics[scale=0.16]{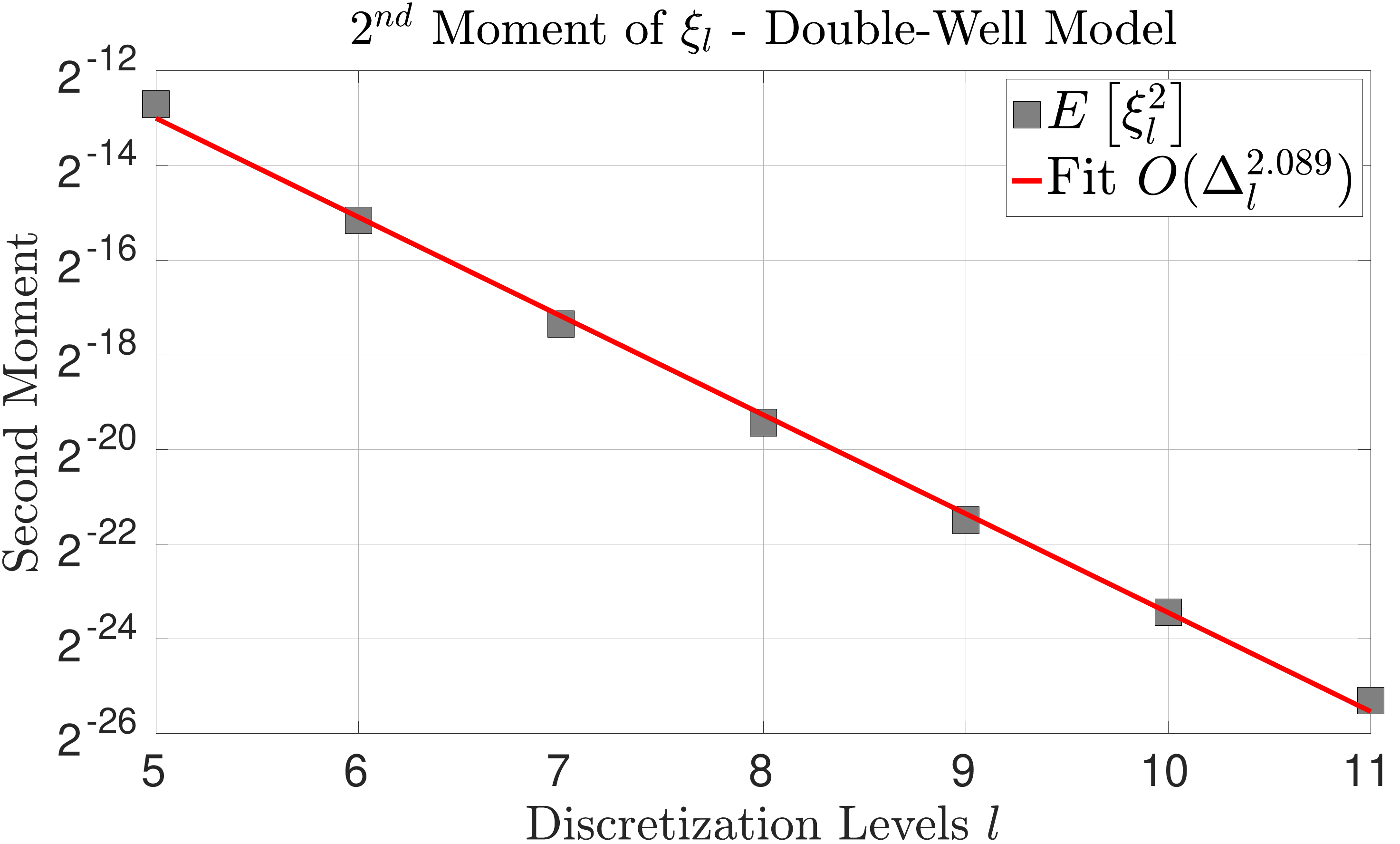}
\includegraphics[scale=0.16]{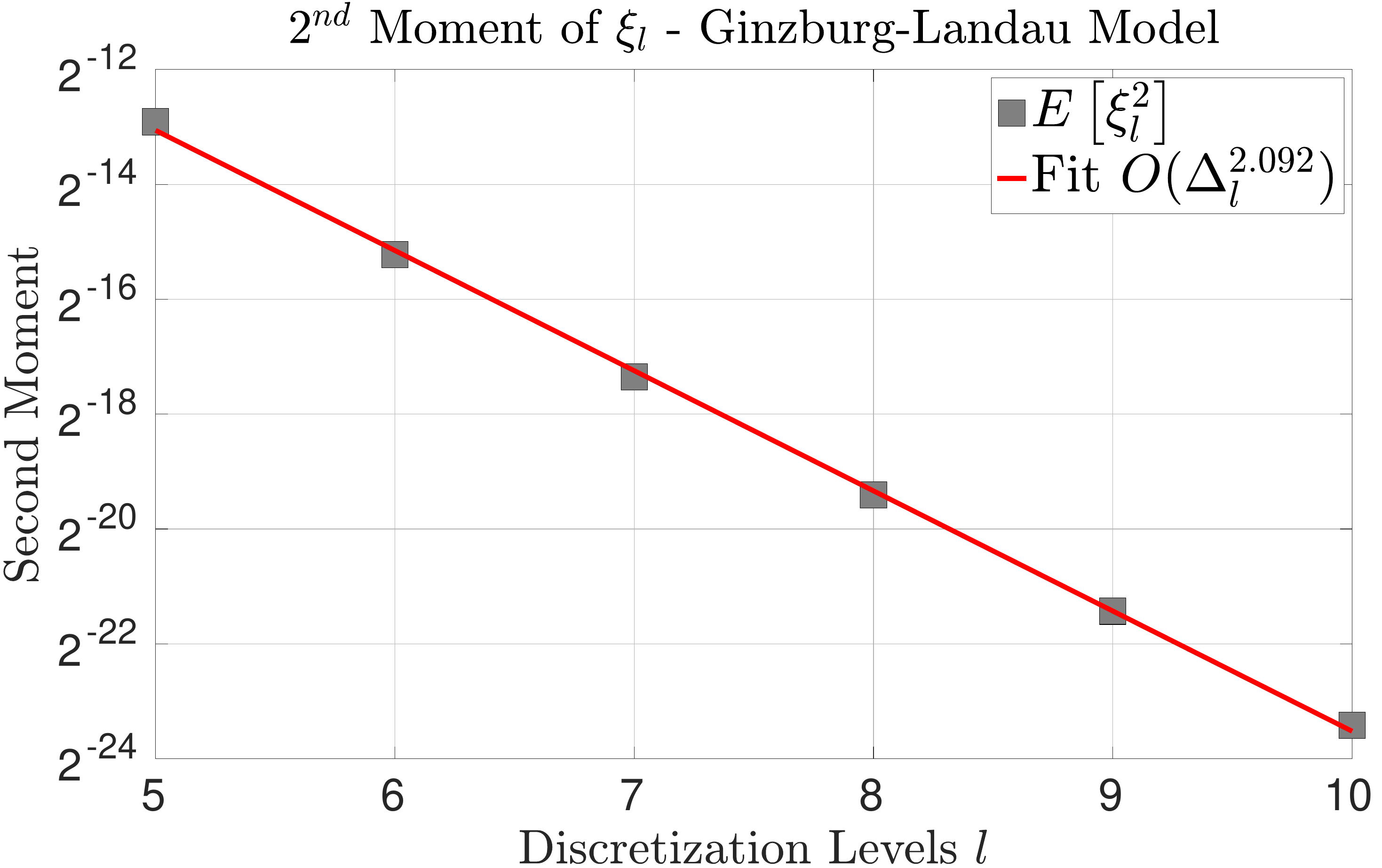}
\captionof{figure}{Second moment of the time-averaged estimator increments $\xi_l = \widehat{[\pi_l-\pi_{l-1}](\varphi)}_{S,T,k,m}$ against discretization level $l$ for the different models in \autoref{sec:numerics}. The function considered here is $\varphi(u) = u_1$, $u\in \mathbb{R}^{2d}$. These plots provide a numerical estimate of the parameter $\beta$ in \eqref{eq:bound_cost} of value 1.}
\label{fig:inc_var}
\end{figure}

\section{Numerical Results}\label{sec:numerics}
\label{sec:num}
In this section we seek to illustrate and verify our theoretical results using three numerical examples arising in statistics and physics. 
Our experiments will be based on demonstrating both finite variance and unbiasedness of the proposed estimator. We will test the methodology on Bayesian logistic regression problem, a double-well potential and Ginzburg-Landau model. Finally, we will compare our methodology to the unbiased Schrodinger--F\"{o}llmer sampler for one particular example.

In all examples below we choose $\kappa$ and $\sigma$ such that $2\kappa = \sigma^2$ and set $b$ in \eqref{eq:disc_dyn2} to $b(x) = -\nabla U(x)$, where $U = -\log \pi$ and $\pi$ is the density of interest that we aim to sample from. The objective is to estimate, unbiasedly, the expectation $\pi(\varphi)$. We find that the algorithm works quite well when $\sigma$ is assigned a large value, but it may collapse for small values. In our simulations, we set $\sigma=3$ and take $\varphi(x) = x$. 

\subsection{Bayesian Logistic Regression}
\label{subsec:BLR}
We consider the binary logistic regression model in which the binary observations $\{Y_i\}_{i=1}^n$ are conditionally independent Bernoulli random variables such that $Y_i \in \{0,1\}$ and
$$
\mathbb{P}(Y_i=1|X_i = x_i, \alpha) = \rho(\alpha^T x_i),
$$
where $\rho: \mathbb{R}\to (0,1)$ defined by $\rho(w) = e^w/(1+e^w)$ is the logistic function and $X_i, ~\alpha \in \mathbb{R}^d$ are the covariates and the unknown regression coefficients, respectively. The prior density for the parameter $\alpha$ is a multivariate normal Gaussian given by
$$
pr(\alpha) = \phi(\alpha;0,\Sigma_\alpha),
$$
where $\Sigma_\alpha$ is defined through its inverse $\Sigma_\alpha^{-1} = \frac{1}{n} (\sum_{i=1}^n X_i X_i^T)$. The covariate vectors $\{X_i\}_{i=1}^n$ are sampled independently from $\mathcal{U}\{-1,1\}^d$ which then are standardized. The density of the posterior distribution of $\alpha$ is then given by
$$
\pi\left(\alpha~|~\{X_i=x_i,Y_i=y_i\}_{i=1}^n \right) \propto \exp\left( - U(\alpha) \right),
$$
where 
$$
U(\alpha) = -\sum_{i=1}^n \left[y_i \alpha^T x_i-\log(1+\exp(\alpha^Tx_i)) \right] +\frac{1}{2} \alpha^T \Sigma_\alpha^{-1} \alpha,
$$
and 
$$
\nabla U(\alpha) = - \sum_{i=1}^n y_i x_i + \sum_{i=1}^n \frac{\exp(\alpha^Tx_i)x_i}{1+\exp(\alpha^Tx_i)} + \alpha^T \Sigma_\alpha^{-1}. 
$$
We take $d=5$ and $n=100$. The reference expectation is the mean of $2.6\times 10^8$ samples computed through the unbiased MCMC method proposed in \cite{jacob1} with a random-walk Metropolis-Hasting kernel.

\subsection{Double-Well Model}
\label{subsec:DW}
For our second example we consider sampling from $\pi(x) = \exp{(-U(x))}$, where $U$ is a double-well potential given by
\begin{equation}
\label{eq:dwm}
U(x) = \frac{1}{4} \|x\|_2^4 - \frac{1}{2}\|x\|_2^2, \qquad x\in \mathbb{R}^d.
\end{equation}
The double-well potential is one of several quartic potentials of substantial importance in quantum mechanics and quantum field theory \cite{double} for the investigation of different physical phenomena or mathematical features. The gradient of the potential is given by 
$$
\nabla U(x) = (\|x\|_2^2 -  1)~x.
$$
In our simulations, we take $d=100$. The true reference expectation is $\pi(\varphi) = 0$.

\subsection{Ginzburg-Landau Model}
\label{subsec:GL}
In the final example we consider the Ginzburg-Landau (GL) model \cite{GL,HK15, GR19}, a model that describes a thermodynamic system that undergoes continuous phase transitions at a temperature $T = T_c$, from a high-temperature symmetric phase to a low-temperature ordered phase in which some symmetry is broken. We denote by the random variable $\psi\in \mathbb{R}$ the order parameter, which is assumed to be spatially dependent, i.e., $\psi=\psi(x)$, $x\in \Omega \subseteq\mathbb{R}^3$, and $\nabla_x \psi \cdot n=0$, where $n$ is a unit vector normal to the boundary $\partial \Omega$. In the absence of external fields, the probability of a fluctuation $\psi(x)$ is given by
$$
\pi(\psi) \propto \exp\left\{-U(\psi(x)) \right\},
$$
with $U$ is the GL free energy functional defined via
$$
U(\psi(x)) = \int_{\Omega} \left[\frac{1-\overline{T}}{2} \psi^2(x)+\frac{\gamma \overline{T}}{2}\|\nabla_x \psi(x)\|^2 + \frac{\zeta\overline{T}}{4} \psi^4(x)\right] dx,
$$

$\overline{T} = T_c/T$ and $\gamma,\zeta>0$. We now consider the discretized GL model on a 3d-lattice with $d = d_0^3$ sites, $(x_{i,j,k})_{i,j,k=1}^{d_0}$, where $d_0\in \mathbb{N}$, with spacing equal to one. Periodic boundary conditions are applied to the system where we have $x_{11,j,k} = x_{1,j,k}$ and $x_{0,j,k}=x_{10,j,k}$ and similar situation for the second and the third coordinates. Each site represents a random variable $\psi_{i,j,k} \in \mathbb{R}$. Set $\bm{\psi}:=(\psi_{i,j,k})_{i,j,k=1}^{d_0} \in \mathbb{R}^{d}$, then an approximation to $\pi$ is given by 
$$
\widetilde{\pi}(\bm{\psi}) \propto \exp\left\{-\widetilde{U}(\bm{\psi}) \right\},
$$
with
\begin{align*}
\widetilde{U}(\bm{\psi}) &= \sum_{i,j,k=1}^{d_0} \Big[\frac{1-\overline{T}}{2} \psi_{i,j,k}^2+\frac{\gamma\overline{T}}{2}\left((\psi_{i+1,j,k} -\psi_{i,j,k} )^2+(\psi_{i,j+1,k} -\psi_{i,j,k})^2+(\psi_{i,j,k+1} -\psi_{i,j,k} )^2\right) + \frac{\zeta\overline{T}}{4} \psi_{i,j,k}^4\Big],
\end{align*}
such that a forward finite difference is used to approximate $\nabla_x \psi$. The functional derivative $\delta U(\psi(x'))/\delta \psi(x)$ \\ is given by
\begin{align}
\frac{\delta U(\psi(x'))}{\delta \psi(x)} &= \int_{\Omega} \Big[ (1-\overline{T}) \psi(x')  \delta(x-x') + \gamma\overline{T} \nabla_{x'} \psi(x') \cdot \nabla_{x'} \delta(x-x') + \zeta\overline{T} \psi^3(x') \delta(x-x') \Big] dx' \nonumber\\
&= (1-\overline{T}) \psi(x) - \gamma\overline{T} \Delta_x \psi(x) + \zeta\overline{T} \psi^3(x), \label{eq:func_der}
\end{align}
where we used the following identities 
\begin{align*}
\frac{\delta}{\delta \psi(x)} \psi(x') &= \delta(x-x'),\\
\int_{\Omega} \nabla_{x'} \psi(x') \cdot \nabla_{x'} \delta(x-x') dx'& = - \int_{\Omega} \Delta_{x'} \psi(x') \delta(x-x') dx' = -\Delta_x \psi(x).
\end{align*}
Here $\delta(x)$ denotes the Dirac delta function. In the last identity we used integration by parts with the fact that $\nabla_x \psi \cdot n=0$ for any normal vector $n$ on the boundary. On the lattice, an approximation to the functional derivative in \eqref{eq:func_der}, which we denote by $\widetilde{\nabla_{\bm{\psi}} U}$, is given by
\begin{align*}
\widetilde{\nabla_{\bm{\psi}} U} &= \Big((1-\overline{T})\psi_{i,j,k} + \gamma\overline{T} \vartheta_{i,j,k} +\zeta\overline{T} \psi_{i,j,k}^3\Big)_{i,j,k=1}^{d_0},
\end{align*}
and
\begin{equation*}
\vartheta_{i,j,k} = \big\{6\psi_{i,j,k}-(\psi_{i+1,j,k}+\psi_{i-1,j,k}+\psi_{i,j+1,k}+\psi_{i,j-1,k}+\psi_{i,j,k+1}+\psi_{i,j,k-1})\big\},
\end{equation*}
where a second-order central finite difference is used to approximate $\Delta_x \psi$. Then $b-\widetilde{\nabla_{\bm{\psi}} U}$ in \eqref{eq:disc_dyn2}. 
The true reference expectation is $\pi(\varphi) = 0\in \mathbb{R}^d$, where $d=10^3$. We also set $\overline{T} =2$, $\gamma = 0.1$ and $\zeta = 0.5$.

\subsection{Simulation Settings}
\label{subsec:sim_sett}
The time-averaged, unbiased estimator
\begin{align}
\label{eq:avg_ub_est}
\left(\widehat{\pi}(\varphi)_{S,T,k,m}\right)_{\text{avg}} := \frac{1}{M}\sum_{i=1}^M \left(\widehat{\pi}(\varphi)_{S,T,k,m}\right)^{(i)},
\end{align}
where $\left(\widehat{\pi}(\varphi)_{S,T,k,m}\right)^{(i)}$ is the estimator presented in \eqref{eq:basic_ub_est_time} in \autoref{alg:final} and the cost of the single-level, time-averaged estimator
\begin{align}
\label{eq:avg_single_est}
\left(\widehat{\pi_L}(\varphi)_{T,k,m}\right)_{\text{avg}} := \frac{1}{M}\sum_{i=1}^M \left(\widehat{\pi_L}(\varphi)_{T,k,m}\right)^{(i)},
\end{align}
where $\left(\widehat{\pi_L}(\varphi)_{T,k,m}\right)^{(i)}$ is presented in \eqref{eq:time_ave} and $L$ is the discretization level. We will compare the cost of both estimators versus the mean-squared errors (MSEs) that are obtained by running 50 independent simulations of each estimator and are given by
$$
\text{MSE}_{ub} = \frac{1}{50} \sum_{j=1}^{50} \left[ \left(\widehat{\pi}(\varphi)_{S,T,k,m}\right)_{\text{avg}}^{(j)} - \pi(\varphi) \right]^2, \quad \text{and}\quad \text{MSE}_{s} = \frac{1}{50} \sum_{i=1}^{50} \left[ \left(\widehat{\pi_L}(\varphi)_{T,k,m}\right)_{\text{avg}}^{(j)} - \pi(\varphi) \right]^2,
$$
where $\pi(\varphi)$ is the reference expectation. Let $\tau$ be the stopping time, we set $k=100$ and $m=\min\{2k,\tau-1\}$, however, if $\tau < k-1$, we set $k = 0.5 ~\tau$.

For a given $\epsilon >0$, the goal is to obtain an MSE of order $\epsilon^2$. Therefore, for the single-level estimator, we set $L=\mathcal{O}(-\log_2(\epsilon))$ and $M=\mathcal{O}(\epsilon^{-2})$. For the unbiased estimator, in practice, one has to truncate the values of $L$. As a result, we set $\mathbb{P}_L(l)\propto 2^{-1.5l} \mathbb{I}_{\{l_*,\cdots,l_{\text{max}}\}}(l)$, where $l_*,l_{\text{max}}\in \mathbb{N}_0$ and $l_*<l_{\text{max}}$. With this choice, the probability of simulating \eqref{eq:disc_dyn2}-\eqref{eq:disc_dyn3} at a high discretization level is very small. In all examples, we take $l_* = 5$, $l_{\text{max}}=12$ and $M=\mathcal{O}(\epsilon^{-2})$ as in the single-level estimator. The cost to compute the single-level estimator is $\mathcal{C}_{\textrm{s}} := 2^{L+1} \sum_{j=1}^M \tau_L^{(j)}=2^{L+1}~M~\max_{j}\{\tau_L^{(j)}\}\leq \mathsf{C} 2^L~M = \mathcal{O}(\epsilon^{-3})$ 
%(\textcolor{red}{this is gonna raise questions from the reviewers because it does not agree with the simulations. The simulations show that the power in the bound on the cost of the single-level is strictly larger than 3 and it is close to 4 in the Ginzburg-Landau model)},
where $\tau_L^{(j)}$ is the stopping time for the $j$-th replicate that is assumed to be bounded by A\ref{ass:2}. The cost of the unbiased estimator is $\mathcal{C}_{\textrm{ub}}:=\sum_{j=1}^M \mathcal{C}^{(j)}$ where 
\begin{align*}
\mathcal{C}^{(j)} = \left\{\begin{array}{lr}
\tau_{L_j} ~ 2^{L_j+1} & \text{ if } L_j = l_*,\\
\check{\tau}_{L_j,L_{j-1}}~ (2^{L_j+1}+2^{L_j}) & \text{ if } L_j > l_*,\\
\end{array}\right. \qquad L_j \sim \mathbb{P}_L.
\end{align*}

\begin{rem}
\label{rem:new}
We note that even though the rate for the single-level estimator should be $\mathcal{O}(\epsilon^{-3})$, what we observe numerically
is different. As we will discuss in the next subsection, the rates obtained are worse, where it seems that we attain a rate of $\mathcal{O}(\epsilon^{-(3+\delta)})$,
for $\delta>0$. This is for each individual numerical example and we expect that this is associated to the drift coefficient in the dynamics, which are in-turn affected by the target.
The regularity of the drift will influence the practical rate of convergence of the discretization scheme and the time convergence of the dynamics.
\end{rem}
%\todo[inline]{Neil: Ajay I have included this remark above, as numerically the rates are not $\mathcal{O}(\epsilon^{-3})$. If you have some further intuition
%on why this phenomenon occurs, can you help with the remark above explaining?}

By \autoref{rem:1}, the expected cost of each replicate is bounded by a constant $\mathsf{C}$ and hence the overall expected cost is $\overline{\mathcal{C}_{\textrm{ub}}} \leq \mathsf{C} M = \mathcal{O}(\epsilon^{-2})$. Notice that in fact the cost of sampling from $\check{Q}_L$ is $2^L$ and the cost of sampling from $\check{P}_L$ is $2^{L}-1$ plus the cost of sampling from a maximal coupling, which we take it to be one, and therefore the cost of sampling from $\check{K}_L=\alpha\check{Q}_L + (1-\alpha)\check{P}_L$ is $2^{L+1}$. 

We will also compare the unbiased SFS algorithm \cite{sfs} with our unbiased ULD sampler \autoref{alg:final}. For a fair comparison, we consider comparing the average machine time (in seconds) that one independent realization of the 50 realizations takes on a workstation of 52 cores. The unbiased SFS estimator is based on the SF diffusion process, 
$$
dX_t = b(X_t,t)dt + dW_t, \quad X_0=0, ~t\in[0,1],
$$
where the drift term has the specific form
$$
b(x,t) = \nabla \log \mathbb{E}[f(x + W_{1-t})] = \frac{\mathbb{E}_{\phi}[\nabla f(x+\sqrt{1-t}Z)]}{\mathbb{E}_{\phi}[ f(x+\sqrt{1-t}Z)]} = \frac{1}{\sqrt{1-t}}\frac{\mathbb{E}_{\phi}[Z f(x+\sqrt{1-t}Z)]}{\mathbb{E}_{\phi}[ f(x+\sqrt{1-t}Z)]},
$$ 
and $f$ corresponds to the analogous of a likelihood function for a standard Gaussian prior, i.e. for $z\in\mathbb{R}^d$, 
$
f(z) = \pi(z)/\phi(z),
$
with $\phi(z)$ the standard $d-$dimensional Gaussian density.  For $l\in\{0,1,\dots\}$ and $k\in\{0,1,\dots,\Delta_l^{-1}-1\}$, we discrete the above SDE as
$$
\widetilde{X}_{(k+1)\Delta_l}^{l,N} = \widetilde{X}_{k\Delta_l}^{l,N} + \hat{b}(\widetilde{X}_{k\Delta_l}^{l,N},k\Delta_l)\Delta_l + W_{(k+1)\Delta_l} - W_{k\Delta_l},
$$ 
where, for $x\in\mathbb{R}^d$
$$
\hat{b}(x,k\Delta_l) = \frac{\frac{1}{N}\sum_{i=1}^N \nabla f(x+\sqrt{1-k\Delta_l}Z^i)}{\frac{1}{N}\sum_{i=1}^N f(x+\sqrt{1-k\Delta_l}Z^i)},
$$
and, for $i\in\{1,\dots,N\}$,  $Z^i\stackrel{\textrm{i.i.d.}}{\sim}\mathcal{N}_d(0,I)$. The idea of the unbiased SFS is based upon double-randomization techniques (e.g. \cite{ub_bip,ub_pf}) where both $N$ and $l$ are chosen at random in a specific manner.

\subsection{Results and Discussion}
\label{subsec:discuss_num}

In \autoref{fig:mse_cost} (a)--(c), we plot the cost of the estimators in \eqref{eq:avg_ub_est}-\eqref{eq:avg_single_est} against their MSE for the models above. The MSE of the proposed unbiased estimator presented in \autoref{alg:final} decays at the optimal rate of $1/\mathcal{C}_{\text{ub}}=\mathcal{O}(\epsilon^2)$ as shown in the plots, which is as expected, and outperforms that of the single-level estimator given in \eqref{eq:avg_single_est}.
As eluded to, from \autoref{rem:new} the actual rate which we obtain for each model is around  $\mathcal{O}(\epsilon^{-(3+\delta)})$, for $\delta>0$, which differs for each model. This is especially the case for the GL model which, compared to the other models, has a more complicated drift term. Therefore we believe this constitutes to a  modified rate for the single-level estimator. Further investigation is needed to verify this, but this is beyond the scope of this work and we leave it for future work. Nonetheless, as stated, the unbiased estimator outperforms the single estimator for each model example.

In \autoref{fig:mse_cost} (d), the cost of our unbiased estimator and that of the SFS, measured in seconds, is plotted against the MSE of each method. As we can see, the new methodology beats the unbiased SFS and in fact the MSE decays faster than expected by the theory. This suggesting to attain similar values of MSE, our unbiased estimator takes considerably less time than that using the SFS, as presented in \cite{sfs}.

\begin{figure}[!h]
\centering
\subfloat[]{
\includegraphics[scale=0.18]{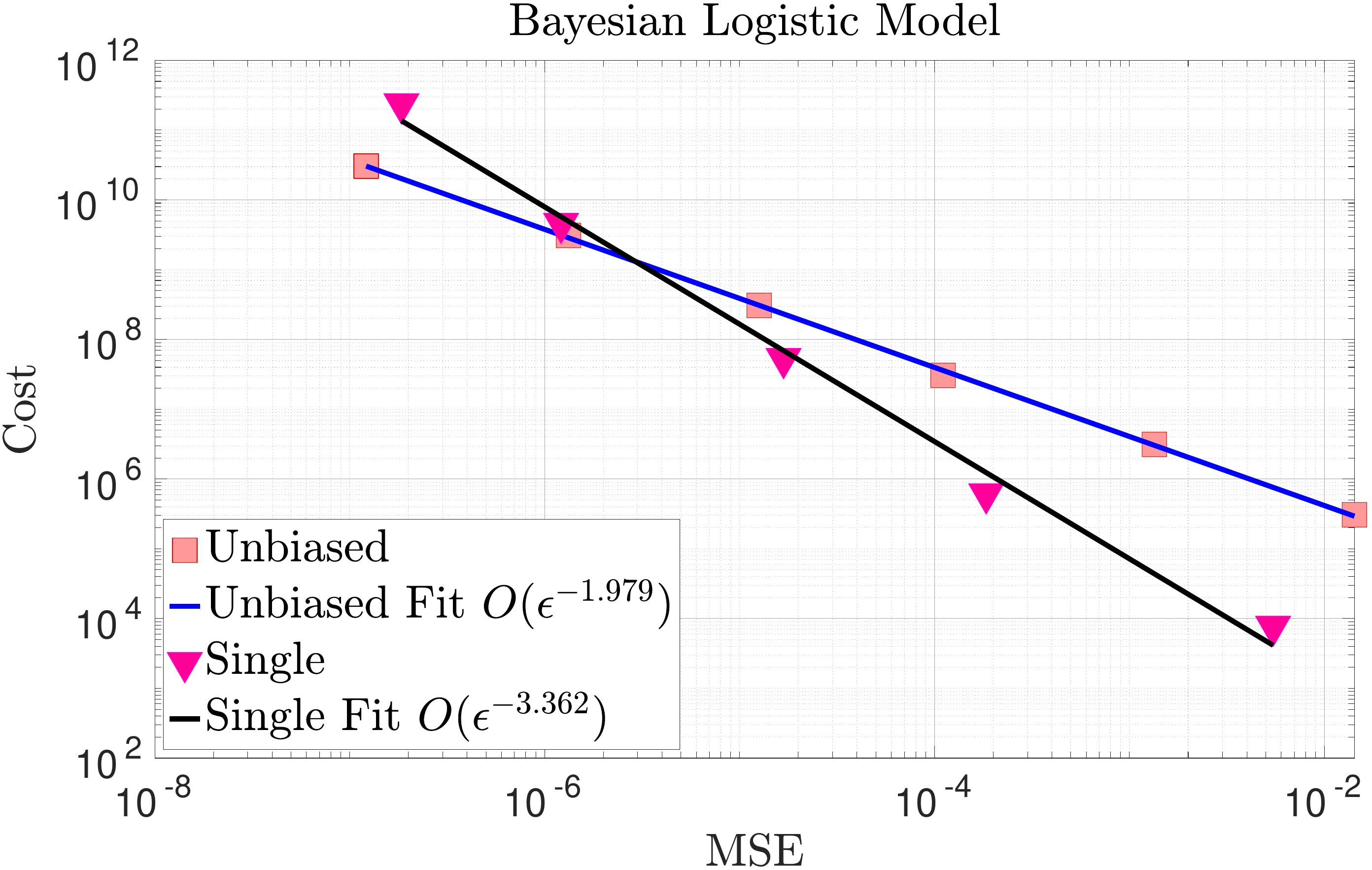}}
\subfloat[]{
\includegraphics[scale=0.18]{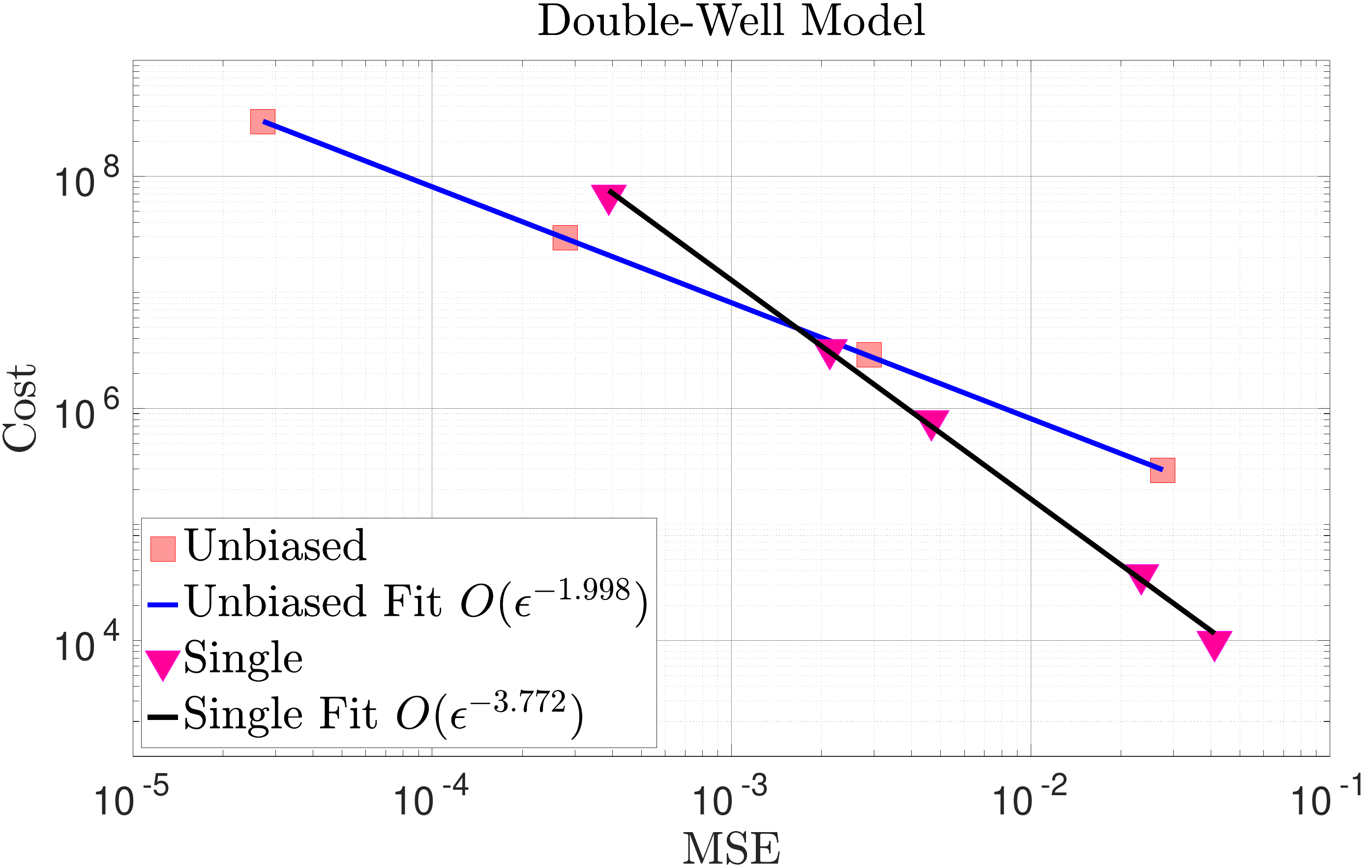}}\\
\subfloat[]{
\includegraphics[scale=0.18]{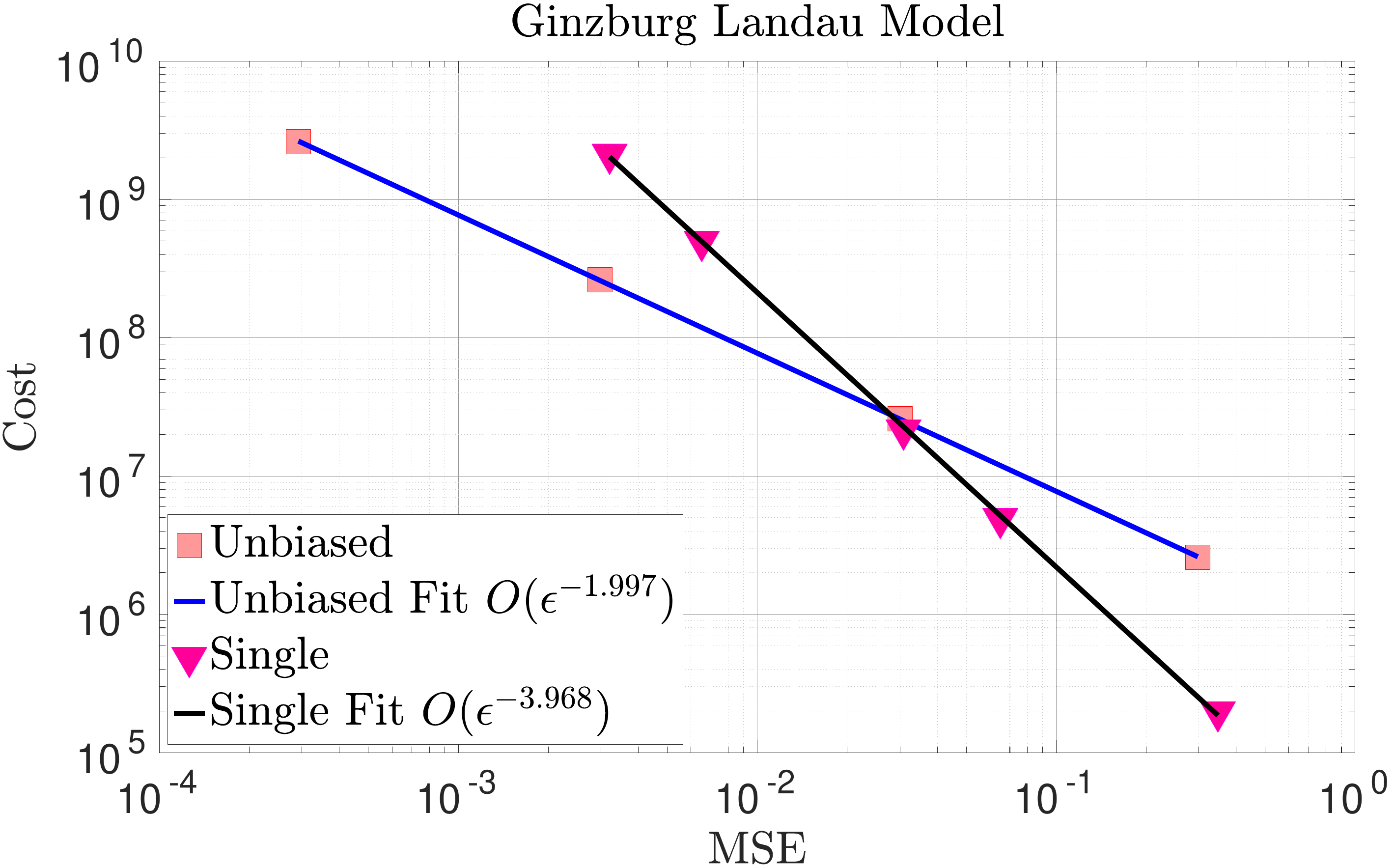}}
\subfloat[]{
\includegraphics[scale=0.18]{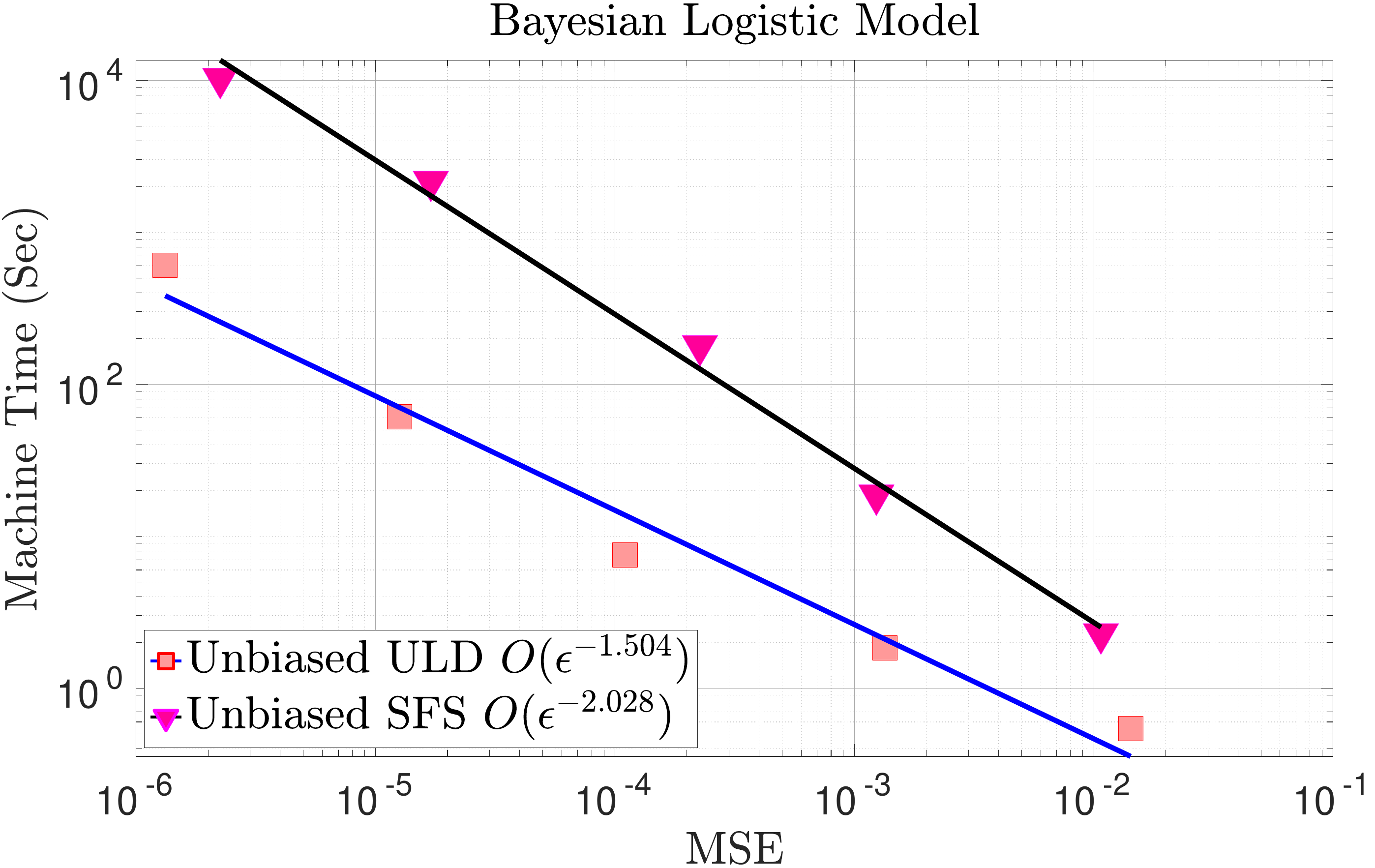}}
\caption{In (a)--(c), we plot the cost versus MSE for the single-level, time-averaged estimator in \eqref{eq:time_ave} and the time-averaged unbiased estimator presented in \eqref{eq:basic_ub_est_time} in \autoref{alg:final}. In (d), we plot the machine time versus MSE for the unbiased estimator in \autoref{alg:final} and the unbiased SFS estimator in \cite{sfs} for the Bayesian logistic model.}
\label{fig:mse_cost}
\end{figure}

\subsection{Comparison to Unbiased MALA}
 {
Our final numerical experiment we present in this work is a comparison of our method to unbiased methodologies proposed by Jacob et al. \cite{jacob1}. In particular we will compare it
to the unbiased MALA, which is a well-known Metropolis Hastings (MH) method based on the following Langevin dynamics
\begin{align*}
dX_t &= - \frac{1}{2}\nabla \log \pi(X_t)dt + \sqrt{2 \delta^{-1}}dB_t,
\end{align*}
with discretization
\begin{align*}
X_{(k+1)\Delta_l} & =  X_{k\Delta_l} -\frac{\Delta_l}{2} \nabla \log \pi(X_{k\Delta_l}) + \sqrt{2 \Delta_l \delta^{-1}}\left(B_{(k+1)\Delta_l}-B_{k\Delta_l}\right),
\end{align*}
where where $\{B_t\}_{t\geq 0}$ is a standard $d-$dimensional Brownian motion and $\delta>0$ denotes the inverse temperature. The acceptance probability associated with it arises from the usual MH-type algorithms. Further details on MALA can be found in the following references \cite{greg,robert,RS03}. Our numerical example will consist
of a comparison between our unbiased scheme and U-MALA, which is tested on the double-well model \eqref{eq:dwm}. The U-MALA was first discussed in the work of Heng et al. \cite{ub_hmc} as a simplified version of the HMC coupling. As a result the coupling associated with U-MALA is much simpler than the U-HMC couplings, which follow from \cite{jacob1}, which also exploit synchronous maximal couplings. Our choice of using this model, is that the density of interest is bimodal, 
which should constitute to a difference in performance over the toy logistic regression model. The dimension is chosen as $d=75$ and again we run 50 independent simulations to compute the MSE. The other parameter choices are consistent with that discussed in Section \ref{subsec:sim_sett}. We also set $\delta=1$ for our experiments. 
\\
We present our simulations in Figure  \ref{fig:new}. As we can observe for both subplots the MSE-to-cost ratio is better, related to both CPU and theoretical cost, for \autoref{alg:final}. The ratio of the rates between each unbiased estimators is consistent in both subplots. As mentioned we believe the reasons for the under-performance of the unbiased MALA proposed by \cite{ub_hmc,jacob1} are due to bimodality of $U(x)$ and the fact
the unbiasedness is only related to the bias coming from the MCMC, not the discretization error, whereas our algorithm aims at removing both the MCMC and discretization bias.}

\begin{figure}[h!]
 \centering
  \includegraphics[width=8.7cm]{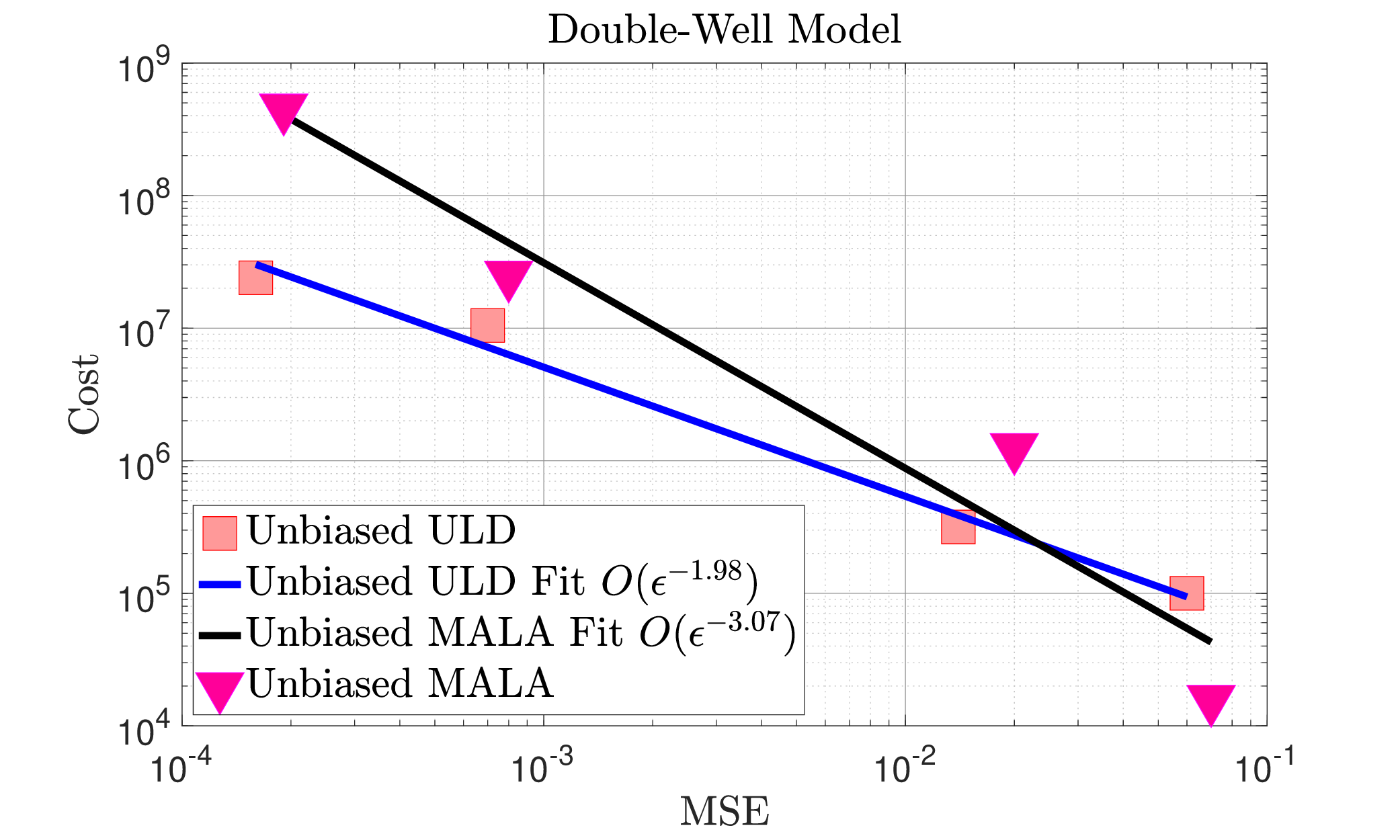}
    \includegraphics[width=8.7cm]{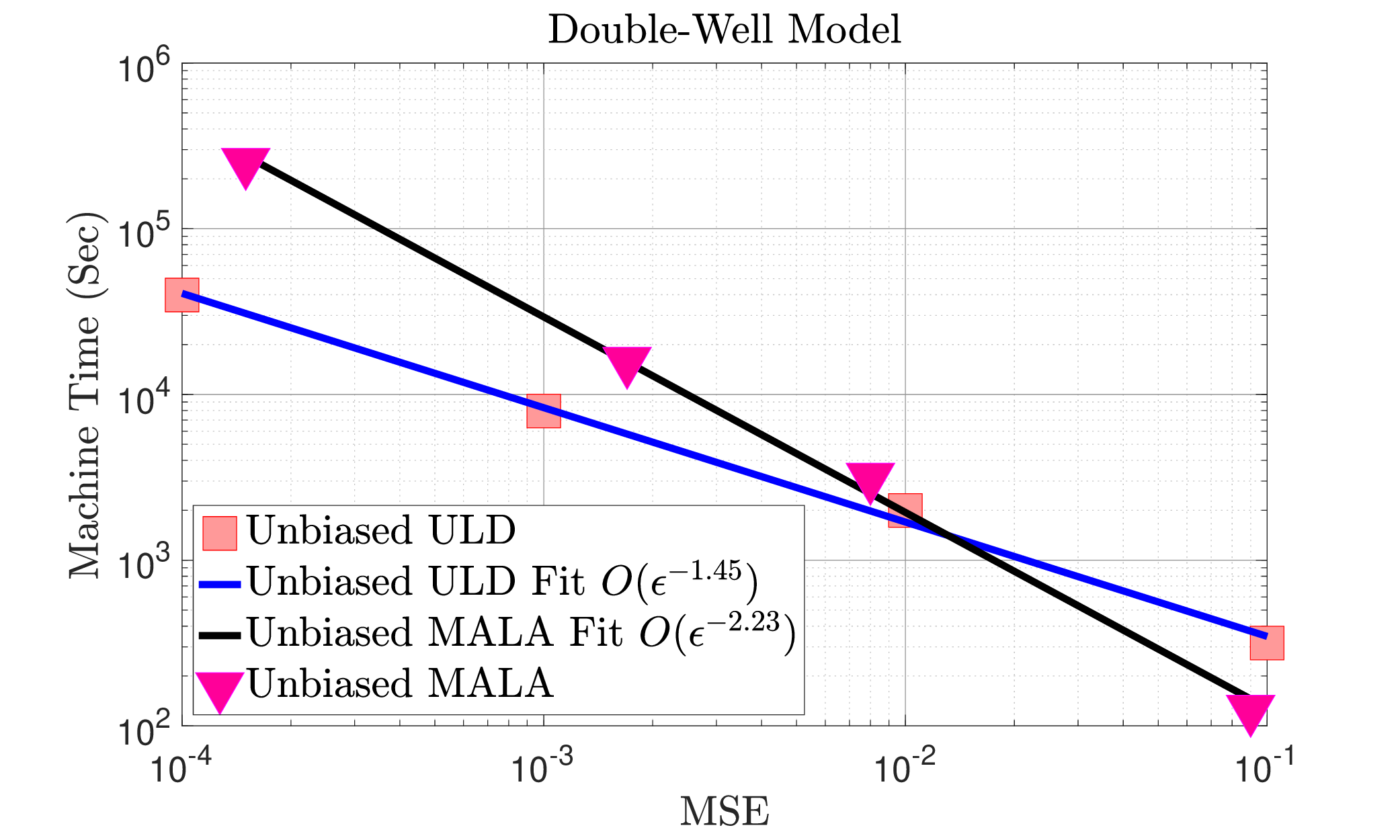}
    \caption{ {This plot compares the time-average unbiased estimator \eqref{eq:basic_ub_est_time} in \autoref{alg:final} to the unbiased MALA. Left: Theoretical cost vs MSE. Right: CPU time cost vs MSE.}}
        \label{fig:new}
\end{figure}

\section{Conclusion}
\label{sec:conc}
Our motivation from this work was to provide a new unbiased estimator, which is based on the discretized underdamped Langevin dynamics (ULD) \eqref{eq:disc_dyn1}-\eqref{eq:disc_dyn2}.
The ULD has sparked recent interest in both the statistics and machine learning community, and as a result we wanted to see if such a dynamics could be used in the context of unbiased estimation.
We introduced a new methodology based on ULD, which was based on the double-randomization schemes used for unbiased estimation. Subsequently we proved that our new estimator is unbiased
 with finite variance under suitable assumptions. To verify our theory, we implemented our methodology on a range of interesting model problems, such as the stochastic Ginzburg-Landau model, the double-well model
and a Bayesian logistic regression problem. We also justified such an
 estimator by comparing it to our known estimators, developed in a similar manner, such as the unbiased Schr\"{o}dinger--F\"{o}llmer sampler presented in \cite{sfs}  {and the unbiased Metropolis adjusted Langevin algorithm presented in \cite{ub_hmc,jacob1}.} 
 
 From this work, there are a number of research directions one can consider. A rather obvious one would be to use the current methodology aimed at unbiased estimation for the both the score function and the Hessian \cite{chada,ub_grad}. This has already been considered previously, but in the context of our methodology which may prove to be more useful. Another direction is to extend the ergodicity result presented in \cite{disc_lange}, where one does not have the requirement that $l_*$ has to be large enough to ensure $\mathcal{V}$-uniform ergodicity. This may prove to be challenging, but from a theoretical perspective would be of interest, especially if one can attain geometric ergodicity. One could also consider the mean-field ULD, which has shown to be promising for deep neural architectures \cite{kazey}, in terms of  the trainability of two-layer neural networks \cite{Chizat,mei}. Related to the above point, another extension could be the perturbed ULD \cite{duncan}, which has demonstrated improvements for sampling even more complex probability measures. \textcolor{black}{A final direction could be providing an extension related to the work of \cite{Leary}. The authors were able to provide an upper bound on one-step meeting probabilities, related to the proposal and the acceptance probability based on particular setup.}

\subsubsection*{Acknowledgements}

All three authors were supported by KAUST baseline funding. We would like to thank the editors and reviewers for their guidance which has greatly improved the article.

\appendix

\section{Proof}

The proof of Proposition \autoref{prop:main_prop} is virtually identical to that of \cite[Theorem 2.1]{disc_model}. There is only place where the proof has to be modified and we
give the result below. Below $C$ is a generic constant which does not depend upon $n,l$ and whose value changes upon each appearance. The expectation operator $\mathbb{E}$
relates to law which generates the simulated process used to compute \eqref{eq:basic_ub_est}.

\begin{lem}
Assume (A\ref{ass:1},\ref{ass:4}). Then there exists a $C\in(0,\infty)$ such that for any $(l,n)\in\mathbb{N}_{l_*}\times\mathbb{Z}^+$ we have:
\begin{eqnarray*}
\mathbb{E}[\mathbb{I}_{\mathsf{B}(C,\Delta_l^{\beta_2},\tilde{\mathsf{d}})^c}(Z_{n,l,l-1})] & \leq & C(n+1)\Delta_l^{\beta_2(2+\epsilon)},\\
\max\left\{\mathbb{E}[\tilde{\mathsf{d}}(U_{n,l},U_{n,l-1})^{2+\epsilon}],\mathbb{E}[\tilde{\mathsf{d}}(\tilde{U}_{n,l},\tilde{U}_{n,l-1})^{2+\epsilon}]\right\} & \leq &
C(n+1)\Delta_l^{\beta_2(1+\tfrac{\epsilon}{2})},
\end{eqnarray*}
where $\beta_2$ and $\epsilon$ are as (A\ref{ass:4}).
\end{lem}

\begin{proof}
The first inequality is proved in an identical manner to \cite[Lemma A.3.]{disc_model}, so we consider only the second inequality. This latter inequality is shown for 
$\mathbb{E}[\tilde{\mathsf{d}}(U_{n,l},U_{n,l-1})^{2+\epsilon}]$, only, as the argument for the other term is the same up to changes in notation. Let $n\in\mathbb{N}$, then we have
$$
\mathbb{E}[\tilde{\mathsf{d}}(U_{n,l},U_{n,l-1})^{2+\epsilon}] = \mathbb{E}[\tilde{\mathsf{d}}(U_{n,l},U_{n,l-1})^{2+\epsilon}\mathbb{I}_{\mathsf{B}(C,\Delta_l^{\beta_2},\tilde{\mathsf{d}})^c}(Z_{n,l,l-1})] + \mathbb{E}[\tilde{\mathsf{d}}(U_{n,l},U_{n,l-1})^{2+\epsilon}\mathbb{I}_{\mathsf{B}(C,\Delta_l^{\beta_2},\tilde{\mathsf{d}})}(Z_{n,l,l-1})].
$$
For the second term on the R.H.S.~as $Z_{n,l,l-1}\in \mathsf{B}(C,\Delta_l^{\beta_2},\tilde{\mathsf{d}})$, we have $\tilde{\mathsf{d}}(U_{n,l},U_{n,l-1})^{2+\epsilon}\leq C\Delta_l^{\beta_2(2+\epsilon)}$, so we focus on the first term on the R.H.S..~Applying Cauchy-Schwarz and the first-statement of the Lemma, we have
$$
\mathbb{E}[\tilde{\mathsf{d}}(U_{n,l},U_{n,l-1})^{2+\epsilon}\mathbb{I}_{\mathsf{B}(C,\Delta_l^{\beta_2},\tilde{\mathsf{d}})^c}(Z_{n,l,l-1})] \leq
\mathbb{E}[\tilde{\mathsf{d}}(U_{n,l},U_{n,l-1})^{2(2+\epsilon)}]^{1/2} C(n+1)\Delta_l^{\beta_2(1+\tfrac{\epsilon}{2})}.
$$
Then using $\tilde{\mathsf{d}}(U_{n,l},U_{n,l-1})^{4(2+\epsilon)}\in\mathcal{L}_{\mathscr{V}\otimes\mathscr{V}}$, we get
$$
\mathbb{E}[\tilde{\mathsf{d}}(U_{n,l},U_{n,l-1})^{2+\epsilon}\mathbb{I}_{\mathsf{B}(C,\Delta_l^{\beta_2},\tilde{\mathsf{d}})^c}(Z_{n,l,l-1})] \leq
\mathbb{E}[\{\mathscr{V}(U_{n,l})\mathscr{V}(U_{n,l-1})\}^{\frac{1}{2}}]^{1/2} C(n+1)\Delta_l^{\beta_2(1+\tfrac{\epsilon}{2})}.
$$
Then applying Cauchy-Schwarz and using (A\ref{ass:1}) 1.~we have
$$
\mathbb{E}[\tilde{\mathsf{d}}(U_{n,l},U_{n,l-1})^{2+\epsilon}\mathbb{I}_{\mathsf{B}(C,\Delta_l^{\beta_2},\tilde{\mathsf{d}})^c}(Z_{n,l,l-1})] \leq
C(n+1)\Delta_l^{\beta_2(1+\tfrac{\epsilon}{2})},
$$
from which the proof can be completed.
\end{proof}

\section{Algorithms}

In this Appendix we provide each algorithm required for our unbiased estimator. This is related to the simulation from the associated kernels
$\check{Q}_l$, $\check{Q}_{l,l-1}$, $\check{P}_l$, $\check{P}_{l,l-1}$, which are discussed in \Crefrange{alg:q_l_sim}{alg:p_l_l-1_sim}.
For the earlier algorithms, fuller details can be found in \cite{thor}. Recall that our objective is to construct
the coupling, or coupled kernels, $\check{K}_l$ and $\check{K}_{l,l-1}$, which can be decomposed with the abovely stated kernels where
\begin{eqnarray*}
\check{K}_l &=& \alpha \check{Q}_l + (1-\alpha)\check{P}_l, \\
\check{K}_{l,l-1}&=&\mathbb{I}_{D^2}(u_l,\tilde{u}_l,u_{l-1},\tilde{u}_{l-1})\check{Q}_{l,l-1} + \mathbb{I}_{(D^2)^c}(u_l,\tilde{u}_l,u_{l-1},\tilde{u}_{l-1})[\alpha\check{Q}_{l,l-1} + (1-\alpha)\check{P}_{l,l-1}],
\end{eqnarray*}
where $\alpha \in (0,1)$. Finally we require the coupling $\overline{K}_{l,l-1}$, which is required for the initialization before sampling from $\check{K}_{l,l-1}$, and is presented in \autoref{alg:k_l_l-1_sim}.
\textcolor{black}{
\begin{center}
%\vspace{-0.3cm}
\captionsetup[algorithm]{style=algori}
\captionof{algorithm}{\textcolor{black}{Sampling from kernel $\check{Q}_l$.}}
\label{alg:q_l_sim}
\begin{enumerate}
\item{{\bf Input}: $l$, $(u,\tilde{u})=\left((x_0,v_0),(\tilde{x}_0,\tilde{v}_0)\right)$.}
\item{Sample $(\Gamma_{0,l},\dots,\Gamma_{\Delta_l^{-1}-1,l})$ and $(B_{\Delta_l}, B_{2\Delta_l}-B_{\Delta_l},\dots, B_1-B_{1-\Delta_l})$.}
\item{Run the recursion \eqref{eq:disc_dyn2}-\eqref{eq:disc_dyn3} with $\{\Gamma_{k,l}\}_{k=0}^{\Delta_l^{-1}-1}$ and $\{B_{(k+1)\Delta_l}-B_{k\Delta_l}\}_{k=0}^{\Delta_l^{-1}-1}$ \\ up-to time 1. %\\ This gives the sequences $\{x_{k\Delta_l},v_{k\Delta_l}\}_{k\in\{1,\dots,\Delta_{l}^{-1}\}}$ and $\{\tilde{x}_{k\Delta_l},\tilde{v}_{k\Delta_l}\}_{k\in\{1,\dots,\Delta_{l}^{-1}\}}$. 
}
\item{{\bf Output}: $u'=(x_1,v_1)$ and $\tilde{u}'=(\tilde{x}_1,\tilde{v}_1)$ as generated in step 3.}
\end{enumerate}
\vspace{-0.21cm}
\hrulefill
\end{center}}

\textcolor{black}{
\begin{center}
%\vspace{-0.3cm}
\captionsetup[algorithm]{style=algori}
\captionof{algorithm}{Sampling from kernel $\check{P}_l$.}
\label{alg:p_l_sim}
\begin{enumerate}
\item{{\bf Input}: $l$, $(u,\tilde{u})=\left((x_0,v_0),(\tilde{x}_0,\tilde{v}_0)\right)$.}
\item{Sample  $(\Gamma_{0,l},\dots,\Gamma_{\Delta_l^{-1}-2,l})$ and $(B_{\Delta_l}, B_{2\Delta_l}-B_{\Delta_l},\dots, B_{1-\Delta_l}-B_{1-2\Delta_l})$.}
\item{Run the recursion \eqref{eq:disc_dyn2}-\eqref{eq:disc_dyn3} with $\{\Gamma_{k,l}\}_{k=0}^{\Delta_l^{-1}-2}$ and $\{B_{(k+1)\Delta_l}-B_{k\Delta_l}\}_{k=0}^{\Delta_l^{-1}-2}$ \\ up-to time $1-\Delta_l$. %\\ This gives two sequences $\{x_{k\Delta_l},v_{k\Delta_l}\}_{k\in\{1,\dots,\Delta_{l}^{-1}-1\}}$ and $\{\tilde{x}_{k\Delta_l},\tilde{v}_{k\Delta_l}\}_{k\in\{1,\dots,\Delta_{l}^{-1}-1\}}$.
}
\item{Sample $\left((X_1,V_1),(\tilde{X}_1,\tilde{V}_1)\right)|\left((x_{1-\Delta_l},v_{1-\Delta_l}),(\tilde{x}_{1-\Delta_l},\tilde{v}_{1-\Delta_l})\right)$ from a maximal coupling of $\displaystyle p_l(x_1,v_1|x_{1-\Delta_l},v_{1-\Delta_l}) \quad\textrm{and}\quad p_l(\tilde{x}_1,\tilde{v}_1|\tilde{x}_{1-\Delta_l},\tilde{v}_{1-\Delta_l}),
$
where $p_l\sim N_{2}(m,C)$ where $m, C$ are determined from  \eqref{eq:disc_dyn2}-\eqref{eq:disc_dyn3}. 
}
\item{{\bf Output}: $u'=(x_1,v_1)$ and $\tilde{u}'=(\tilde{x}_1,\tilde{v}_1)$ as generated in step 4..}
\end{enumerate}
\vspace{-0.2cm}
\hrulefill
\end{center}}

\textcolor{black}{
\begin{center}
\captionsetup[algorithm]{style=algori}
\captionof{algorithm}{Sampling from coupled kernel $\overline{K}_{l,l-1}$.}
\label{alg:k_l_l-1_sim}
\begin{enumerate}
\item{{\bf Input}: $l\in\{l_*+1,l_*+2,\dots\}$, and  $(u_l,u_{l-1})=\left((x_{0,l},v_{0,l}),(x_{0,l-1},v_{0,l-1})\right)$.}
\item{Sample  $(\Gamma_{0,l},\dots,\Gamma_{\Delta_l^{-1}-1,l})$ and  $(B_{\Delta_l}, B_{2\Delta_l}-B_{\Delta_l},\dots, B_1-B_{1-\Delta_l})$.
\\ Concatenate to obtain $(\Gamma_{0,l-1},\dots,\Gamma_{\Delta_{l-1}^{-1}-1,l-1})$ \\
and $(B_{\Delta_{l-1}}, B_{2\Delta_{l-1}}-B_{\Delta_{l-1}l},\dots, B_1-B_{1-\Delta_{l-1}})$.}
\item{For $s\in\{l,l-1\}$: run the recursion \eqref{eq:disc_dyn2}-\eqref{eq:disc_dyn3} with $\{\Gamma_{k,s}\}_{k=0}^{\Delta_s^{-1}-1}$ \\ and $\{B_{(k+1)\Delta_s}-B_{k\Delta_s}\}_{k=0}^{\Delta_s^{-1}-1}$  up-to time 1. % \\ This gives one sequence $\{x_{k\Delta_s,s},v_{k\Delta_s,s}\}_{k\in\{1,\dots,\Delta_{s}^{-1}\}}$. 
}
\item{{\bf Output}: $(u_l,u_{l-1})=\left((x_{1,l},v_{1,l}),(x_{1,l-1},v_{1,l-1})\right)$ as generated in step 3..}
\end{enumerate}
\vspace{-0.2cm}
\hrulefill
\end{center}}

\textcolor{black}{
\begin{center}
\captionsetup[algorithm]{style=algori}
\captionof{algorithm}{Sampling from coupled kernel $\check{Q}_{l,l-1}$.}
\label{alg:q_l_l-1_sim}
\begin{enumerate}
\item{{\bf Input}: $l\in\{l_*+1,l_*+2,\dots\}$, and for $s\in\{l,l-1\}$,  \\$(u_s,\tilde{u}_{s})=\left((x_{0,s},v_{0,s}),(\tilde{x}_{0,s},\tilde{v}_{0,s})\right)$.}
\item{Sample $(\Gamma_{0,l},\dots,\Gamma_{\Delta_l^{-1}-1,l})$ and $(B_{\Delta_l}, B_{2\Delta_l}-B_{\Delta_l},\dots, B_1-B_{1-\Delta_l})$.
\\ Concatenate to obtain $(\Gamma_{0,l-1},\dots,\Gamma_{\Delta_{l-1}^{-1}-1,l-1})$ \\ and $(B_{\Delta_{l-1}}, B_{2\Delta_{l-1}}-B_{\Delta_{l-1}l},\dots, B_1-B_{1-\Delta_{l-1}})$.}
\item{ For $s\in\{l,l-1\}$: run the recursion \eqref{eq:disc_dyn2}-\eqref{eq:disc_dyn3} with $\{\Gamma_{k,s}\}_{k\in 0}^{\Delta_s^{-1}-1}$ \\ and $\{B_{(k+1)\Delta_s}-B_{k\Delta_s}\}_{k=0}^{\Delta_s^{-1}-1}$ up-to time 1.% \\ This gives two sequences $\{x_{k\Delta_s,s},v_{k\Delta_s,s}\}_{k\in\{1,\dots,\Delta_{s}^{-1}\}}$ and $\{\tilde{x}_{k\Delta_s,s},\tilde{v}_{k\Delta_s,s}\}_{k\in\{1,\dots,\Delta_{s}^{-1}\}}$. 
}
\item{{\bf Output}: for $s\in\{l,l-1\}$, $(u_s,\tilde{u}_{s})=\left((x_{1,s},v_{1,s}),(\tilde{x}_{1,s},\tilde{v}_{1,s})\right)$ as generated in step 3..}
\end{enumerate}
\vspace{-0.2cm}
\hrulefill
\end{center}}
\textcolor{black}{
\begin{center}
\captionsetup[algorithm]{style=algori}
\captionof{algorithm}{Sampling from coupled kernel $\check{P}_{l,l-1}$.}
\label{alg:p_l_l-1_sim}
\raggedright
\begin{enumerate}
\item{{\bf Input}: $l\in\{l_*+1,l_*+2,\dots\}$, and for $s\in\{l,l-1\}$, \\ $(u_s,\tilde{u}_{s})=\left((x_{0,s},v_{0,s}),(\tilde{x}_{0,s},\tilde{v}_{0,s})\right)$.}
\item{Sample $(\Gamma_{0,l},\dots,\Gamma_{\Delta_l^{-1}-2,l})$ and $(B_{\Delta_l}, B_{2\Delta_l}-B_{\Delta_l},\dots, B_{1-\Delta_l}-B_{1-2\Delta_l})$.
\\ Concatenate to obtain $(B_{\Delta_{l-1}}, B_{2\Delta_{l-1}}-B_{\Delta_{l-1}l},\dots, B_{1-\Delta_{l-1}}-B_{1-2\Delta_{l-1}})$  \\ and 
$(\Gamma_{0,l-1},\dots,\Gamma_{\Delta_{l-1}^{-1}-2,l-1})$.}
\item{For $s\in\{l,l-1\}$: run the recursion \eqref{eq:disc_dyn2}-\eqref{eq:disc_dyn3} with $\{\Gamma_{k,s}\}_{k=0}^{\Delta_s^{-1}-1}$ \\ and $\{B_{(k+1)\Delta_s}-B_{k\Delta_s}\}_{k=0}^{\Delta_s^{-1}-1}$  up-to time $1-\Delta_s$. % \\ This gives two sequences $\{x_{k\Delta_s,s},v_{k\Delta_s,s}\}_{k\in\{1,\dots,\Delta_{s}^{-1}-1\}}$ and $\{\tilde{x}_{k\Delta_s,s},\tilde{v}_{k\Delta_s,s}\}_{k\in\{1,\dots,\Delta_{s}^{-1}-1\}}$. 
}
\item{Sample 
{\footnotesize
\begin{eqnarray*}
\left((X_{1,l},V_{1,l}),(\tilde{X}_{1,l},\tilde{V}_{1,l})\right),
\left((X_{1,l-1},V_{1,l-1}),(\tilde{X}_{1,l-1},\tilde{V}_{1,l-1})\right)\Big|\left((x_{1-\Delta_l,l},v_{1-\Delta_l,l}),(\tilde{x}_{1-\Delta_l,l},\tilde{v}_{1-\Delta_l,l})\right), \\  
((x_{1-\Delta_{l-1},l-1},v_{1-\Delta_{l-1},l-1}),\\  (\tilde{x}_{1-\Delta_{l-1},l-1},\tilde{v}_{1-\Delta_{l-1},l-1}))
\end{eqnarray*}
}
from the synchronous pairwise reflection maximal coupling \cite{Bou}. %of
%$$
%p_l(x_{1,l},v_{1,l}|x_{1-\Delta_l,l},v_{1-\Delta_l,l}),  p_l(\tilde{x}_{1,l},\tilde{v}_{1,l}|\tilde{x}_{1-\Delta_l,l},\tilde{v}_{1-\Delta_l,l}), \quad\textrm{and}
%$$
%$$
%p_{l-1}(x_{1,l-1},v_{1,l-1}|x_{1-\Delta_{l-1},l-1},v_{1-\Delta_{l-1},l-1}),  p_{l-1}(\tilde{x}_{1,l-1},\tilde{v}_{1,l-1}|\tilde{x}_{1-\Delta_{l-1},l-1},\tilde{v}_{1-\Delta_{l-1},l-1}),
%$$
%where the density, for $s\in\{l,l-1\}$, $p_s$ is that of a $2d-$dimensional Gaussian distribution whose mean and variance can be determined from one step of \eqref{eq:disc_dyn2}-\eqref{eq:disc_dyn3}.
}
\item{{\bf Output}: for $s\in\{l,l-1\}$, $(u_s,\tilde{u}_{s})=\left((x_{1,s},v_{1,s}),(\tilde{x}_{1,s},\tilde{v}_{1,s})\right)$ as generated in step 4..}
\end{enumerate}
\vspace{-0.2cm}
\hrulefill
\end{center}}


\begin{thebibliography}{99}

\bibitem{benven}
{\sc Benveniste}, A., {\sc Metivier}, N. \& {\sc Priouret}, P.~(1990).
\emph{Adaptive Algorithms and Stochastic Approximations. Applications of Mathematics}.
Spinger: New York.


\bibitem{beskos}
{\sc Beskos}, A., {\sc Papaspiliopoulos}, O., {\sc Roberts}, G. O. \& {\sc Fearnhead}, P.~(2006)
Exact and computationally efficient likelihood-based estimation for discretely observed diffusion processes. 
J. R. Stat. Soc. Ser. B Stat. Methodol., \textbf{68}, pp. 333--382.


\bibitem{beskos2}
{\sc Beskos, A. \& Roberts, G}. (2005).
Exact simulation of diffusions, {\it Ann. Appl. Probab.}, \textbf{15}, 422--2444.

\bibitem{Bou}
{\sc Bou-Rabee, N., Eberle, A., \& Zimmer, R.}~(2020). Coupling and convergence for Hamiltonian Monte Carlo. {\it Ann. Appl. Probab.}, \textbf{30}, 1209--1250.


\bibitem{chada2}
{\sc Chada}, N. K., {\sc Franks}, J., {\sc Jasra}, A., {\sc Law}, K. J. H. \& {\sc Vihola}, M.~(2021)
Unbiased estimation of discretely observed hidden Markov models.
\emph{SIAM/ASA JUQ}, 9(2), 763-787.

\bibitem{chada}
{\sc Chada}, N. K., {\sc Jasra}, A. \& {\sc Yu}, F.~(2022).
Unbiased estimation of the Hessian for partially observed diffusions.
\emph{Proceedings of the Royal Society A.,} {\bf 478}, 20210710.

\bibitem{cheng}
{\sc Cheng}, X., {\sc Chatterji}, N. S. {\sc Bartlett}, P. L., \& {\sc Jordan}, M. I.~(2018).
Underdamped Langevin MCMC: A non-asymptotic analysis.
\emph{Proceedings of Machine Learning Research},  {\bf 75}, 1--24.

\bibitem{Chizat}
{\sc Chizat}, L. \& {\sc Bach} F.~(2018).
On the global convergence of gradient descent for overparameterized models using optimal transport.
Advances in neural information processing systems, 3040--3050.

\bibitem{dalalyan}
{\sc Dalalyan}, A. S. ~(2017).
Theoretical guarantees for approximate sampling from smooth and logconcave densities. 
{\it J.~R.~Statist. Soc. Ser. B}, \textbf{79} 651--676.

\bibitem{duncan}
{\sc Duncan, A., Nusken, N. \& Pavliotis, G. A.} (2017).
Using perturbed underdamped Langevin dynamics to efficiently sample from probability distributions.
{\it J Stat Phys.,} \textbf{169}, 1098--1131.

\bibitem{disc_lange}
{\sc Durmus}, A., {\sc Enfroy}, A. {\sc Moulines}, E., \& {\sc Stoltz}, G.~(2021). Uniform minorization condition and convergence bounds for discretizations of kinetic Langevin dynamics. arXiv preprint.

\bibitem{durmus}
{\sc Durmus, A. \& Moulines, E.} (2016)
Sampling from strongly log-concave distributions with the
Unadjusted Langevin Algorithm. arxiv preprint.

\bibitem{frenkel}
{\sc Frenkel}, D. \& {\sc Smit}, B.~(2002). \emph{Understanding Molecular Simulation, From Algorithms to Applications}. Academic Press: New York.


\bibitem{gao}
{\sc Gao, X., Gurbuzbalaban, M. \& Lingjiong, Z.} (2020).
Breaking reversibility accelerates langevin dynamics for non-convex optimization.
{\it NIPS'20: Proceedings of the 34th International Conference on Neural Information Processing Systems}, \textbf{498}, 17850--17862.


\bibitem{GL}
{\sc Ginzburg, V.L. \& Landau, L.D.}~(1950). {\it Zh. Eksp. Teor. Fiz.} {\bf 20}, 1064. English translation in: L. D. Landau, Collected papers (Oxford: Pergamon Press, 1965) p. 546


\bibitem{glynn2}
{\sc Glynn}, P.~W. \& 
{\sc Rhee}, C.~H.~(2014). 
Exact estimation for Markov chain equilibrium expectations. 
\emph{J. Appl. Probab.}, {\bf 51}, 
377--389.


\bibitem{ub_grad}
{\sc Heng}, J., {\sc Houssineau}, J. \& {\sc Jasra}, A.~(2021). On unbiased score estimation for partially observed diffusions. arXiv preprint.
 {
\bibitem{ub_hmc}
{\sc Heng}, J. \& {\sc Jacob}, P.~(2019). Unbiased Hamiltonian Monte Carlo with couplings. \emph{Biometrika.},106(2):287--302.}



\bibitem{disc_model}
{\sc Heng}, J., {\sc Jasra}, A. {\sc Law}, K. J. H. \& {\sc Tarakanov}, A.~(2023). On unbiased estimation for discretized models. \emph{SIAM/ASA JUQ} (to appear). 

\bibitem{HK15}
{\sc Hohenberg, P.C. \& Krekhov, A.P.}~(2015). An introduction to the Ginzburg-Landau theory of phase transitions and nonequilibrium patterns. \emph{Physics Reports}, {\bf 572}, 1--42.


\bibitem{jacob1}
{\sc Jacob}, P., {\sc O' Leary}, J. \& {\sc Atchad{\'e}}, Y.~(2020). Unbiased Markov chain Monte Carlo with couplings (with discussion).
\emph{J.~R.~Statist. Soc. Ser. B}, {\bf 82}, 543--600.

\bibitem{ub_bip}
{\sc Jasra}, A., {\sc Law}, K. J. H. \& {\sc Lu}, D.~(2021).
Unbiased estimation of the gradient of the log-likelihood in inverse problems. 
\emph{Stat. Comp.} {\bf 31}, 21.

\bibitem{ub_pf}
{\sc Jasra}, A., {\sc Law}, K. J. H. \& {\sc Yu}, F.~(2022). Unbiased filtering of a class of partially observed diffusions.
\emph{Adv. Appl. Probab.} {\bf 54}, 661-687.

\bibitem{mc_sim}
{\sc Jasra}, A., {\sc Law}, K. J. H. \& {\sc Xu}, Y. (2021). Markov chain Simulation for Multilevel Monte Carlo.
\emph{Found. Data Sci.} {\bf 3}, 27-47.

\bibitem{double}
{\sc Jelic}, V. \& and {\sc Marsiglio}, F.~(2012)
The double-well potential in quantum mechanics: a simple, numerically exact formulation.
\emph{Eur. J. Phys}. 33 1651.


\bibitem{kazey}
{\sc Kazeykina}, A., {\sc Ren}, Z. {\sc Tan}, X. \& {\sc Yang}, J.~(2020)
Ergodicity of the underdamped mean-field Langevin dynamics.arXiv preprint.

\bibitem{GR19}
{\sc Livingstone, S., Faulkner, M. F. \& Roberts, G. O.}~(2019). Kinetic energy choice in Hamiltonian/hybrid Monte Carlo. \emph{Biometrika}, {\bf 106}, 2, 303--319.


\bibitem{mei} 
{\sc Mei, S., A. Montanari, A, \& P.-M. Nguyen P.-M.}~(2018), 
A mean field view of the landscape of two-layer neural networks, {\it Proceedings of the National Academy of Sciences}, \textbf{115}, 7665--7671.


\bibitem{mcl}
{\sc McLeish}, D.~(2011). A general method for debiasing a Monte Carlo estimator. \emph{Monte Carlo Meth. Appl.}, {\bf 17}, 301--315.

\bibitem{midd}
{\sc Middleton}, L., {\sc Deligiannidis}, G., {\sc Doucet}, A.~\&{\sc Jacob}, P.~(2019). Unbiased smoothing using particle independent Metropolis-Hastings. In \emph{Proc.
Mach. Learn. Res.}, {\bf 89}.


\bibitem{neal}
{\sc Neal}, R. M.~(2011).
MCMC using Hamiltonian dynamics. 
\emph{Handbook of Markov Chain Monte Carlo.} CRC Press.
 {
\bibitem{greg}
{\sc Pavliotis, G. A. (2014).}
{\it Stochastic Processes and Applications.}
Springer.}

\bibitem{propp}
{\sc Propp}, J. G., \& {\sc Wilson}, D. B.~(1996). Exact sampling with coupled Markov chains and applications to statistical mechanics. \emph{Rand. Struct. Algs.}, {\bf 9},
223--252.



\bibitem{rhee2}
{\sc Rhee}, C. H. (2013). Unbiased estimation with biased samples. Ph.D. thesis, Stanford University.

\bibitem{rhee}
{\sc Rhee}, C. H. \& {\sc Glynn}, P.~(2015). Unbiased estimation with square root convergence for SDE models. \emph{Op. Res.},~{\bf 63}, 1026--1043. 

\bibitem{robert}
{\sc Robert}, C. \& {\sc Casella}, G.~(2004). \emph{Monte Carlo Statistical Methods}. Springer: New York.
 {
\bibitem{RS03}
{\sc Roberts}, G. O. \&  {\sc Stramer}, O.~(2003) 
Langevin diffusions and Metropolis-Hastings algorithms. 
\newblock{\em Methodol. Comput. Appl. Probab.}, 4(4):337--357.}


\bibitem{sfs}
{\sc Ruzayqat, H., Beskos, A., Crisan, D., Jasra, A. \& Kantas, N.}~(2023). Unbiased Estimation using a Class of Diffusion Processes. \emph{J. Comp. Phys.}, {\bf 472}, 111643

\bibitem{levy}
{\sc Ruzayqat, H. \& Jasra A.}~(2022). Unbiased Parameter Inference for a Class of Partially Observed L\'{e}vy-Process Models. \emph{Found. Data Sci.}, \textbf{4}(2), 299--322.

\bibitem{zakai}
 {\sc Ruzayqat, H. \& Jasra A.}~(2020) Unbiased estimation of the solution to Zakai's equation. \emph{Monte Carlo Meth. Appl}., \textbf{26}(2), 113--129.
 
\bibitem{fractional}
{\sc simsekli, U., Lingjiong, Z., Why Teh, Y. \& Gurbuzbalaban, M.} (2020)
Fractional underdamped Langevin dynamics: retargeting SGD with momentum under heavy-tailed gradient noise.
 {\it Proceedings of the 37th International Conference on Machine Learning}, \textbf{832}, 8970--8980.


\bibitem{thor}
{\sc Thorisson}, H. (2000). \emph{Coupling, Stationarity and Regeneration}. Springer: New York.

\bibitem{boom}
{\sc van den Boom}, W., {\sc Jasra}, A., {\sc De Iorio}, M., \& {\sc Beskos}, A.~(2022). Unbiased approximation of posteriors via coupled particle Markov chain Monte Carlo.
\emph{Stat. Comp.} {\bf 32}, article 36. 


\bibitem{vihola}
{\sc Vihola}, M.~(2018).  Unbiased estimators and multilevel Monte Carlo. \emph{Operations Research}, {\bf 66}(2), 448--462.


\bibitem{Leary}
\textcolor{black}{{\sc Wang}, G., {\sc O'Leary}, J. \& {\sc Jacob}, P.~(2021).
Maximal Couplings of the Metropolis-Hastings Algorithm.
\newblock{\em International Conference on Artificial Intelligence and Statistics}, 1225--1233.}

\end{thebibliography}
\end{document}